\documentclass{article}
\usepackage{arxiv}
\usepackage{titling}
\usepackage{titlesec}
\usepackage[english]{babel}
\usepackage[utf8]{inputenc}
\usepackage[T1]{fontenc}

\let\startlocaldefs\relax
\let\endlocaldefs\relax

\input{packages}
\startlocaldefs
\theoremstyle{plain}

\newtheorem{theorem}{Theorem}[section]
\newtheorem{lemma}[theorem]{Lemma}
\newtheorem{proposition}[theorem]{Proposition}
\newtheorem{corollary}[theorem]{Corollary}
\theoremstyle{plain}
\newtheorem{definition}[theorem]{Definition}

\endlocaldefs

\newcommand{\FT}{\mathrm{FT}}
\newcommand{\dd}{\mathrm d}

\newcommand{\NN}{\mathbb N}
\newcommand{\mutilde}{\tilde \mu}
\def\ind{\stackrel{\mbox{\scriptsize{ind}}}{\sim}}
\def\iid{\stackrel{\mbox{\scriptsize{iid}}}{\sim}}
\def\distreq{\stackrel{\mbox{\scriptsize{d}}}{=}}
\newcommand{\R}{\mathbb R}
\newcommand{\E}{\mathbb E}

\newcommand{\useridx}{n}
\newcommand{\usertotal}{N}
\newcommand{\dayidx}{d}
\newcommand{\daytotal}{D}
\newcommand{\daycount}{A_{\dayidx,\useridx}}
\newcommand{\news}[2]{U_{#1}^{(#2)}}
\newcommand{\hatnews}[2]{\hat{U}_{#1}^{(#2)}}
\newcommand{\newstot}[2]{T_{#1}^{(#2)}}
\newcommand{\hatnewstot}[2]{\hat{T}_{#1}^{(#2)}}
\def\pilotdays{D_{0}}
\def\followupdays{D_{1}}
\newcommand{\freq}{j}

\newcommand{\tail}{\alpha}

\newcommand{\prob}{\mathrm{Pr}}
\renewcommand{\mid}{\,|\,}

\def\tgmodel{{\rm TG-SSP}\xspace}
\def\bemodel{{\rm Be-SSP}\xspace}
\def\nbmodel{{\rm NB-SSP}\xspace}

\DeclareMathOperator*{\argmax}{arg\,max}
\DeclareMathOperator*{\argmin}{arg\,min}

\newcommand{\NexpKtwentyOne}{186}
\newcommand{\NexpKfifty}{157}
\newcommand{\NexpKhundred}{107}
\newcommand{\DtotalRees}{213}

\newcommand{\htMAEAsosNBSSPA}{4.8}
\newcommand{\htMAEAsosTGSSPA}{6.2}
\newcommand{\htMAEAsosIBPA}{9.0}
\newcommand{\htMAEAsosLinA}{6.4}
\newcommand{\htMAEAsosJackA}{6.4}
\newcommand{\htMAEAsosNBSSPB}{12.4}
\newcommand{\htMAEAsosTGSSPB}{15.7}
\newcommand{\htMAEAsosIBPB}{16.6}
\newcommand{\htMAEAsosLinB}{17.6}
\newcommand{\htMAEAsosJackB}{16.7}
\newcommand{\htMAEAsosNBSSPC}{11.6}
\newcommand{\htMAEAsosTGSSPC}{21.3}
\newcommand{\htMAEAsosIBPC}{18.7}
\newcommand{\htMAEAsosLinC}{22.8}
\newcommand{\htMAEAsosJackC}{21.3}
\newcommand{\htNAsosA}{128}
\newcommand{\htNAsosB}{93}
\newcommand{\htNAsosC}{38}

\newcommand{\htReductionAsosMin}{21}
\newcommand{\htReductionAsosMax}{38}

\newcommand{\htMAEReesNBSSPA}{0.4}
\newcommand{\htMAEReesTGSSPA}{0.4}
\newcommand{\htMAEReesIBPA}{0.8}
\newcommand{\htMAEReesLinA}{0.4}
\newcommand{\htMAEReesJackA}{0.4}
\newcommand{\htMAEReesNBSSPB}{1.0}
\newcommand{\htMAEReesTGSSPB}{1.0}
\newcommand{\htMAEReesIBPB}{1.3}
\newcommand{\htMAEReesLinB}{1.2}
\newcommand{\htMAEReesJackB}{1.0}
\newcommand{\htMAEReesNBSSPC}{2.4}
\newcommand{\htMAEReesTGSSPC}{2.5}
\newcommand{\htMAEReesIBPC}{2.8}
\newcommand{\htMAEReesLinC}{2.7}
\newcommand{\htMAEReesJackC}{2.5}
\newcommand{\htNReesA}{107}
\newcommand{\htNReesB}{107}
\newcommand{\htNReesC}{107}

\usepackage{authblk}
\pretitle{\begin{center}\Large}
\posttitle{\par\end{center}\vskip 0.5em}
\date{}
\title{\mbox{\fontsize{14}{17}\selectfont\bfseries Online activity prediction via generalized Indian buffet process models}}
\author{Mario Beraha\textsuperscript{1}, Lorenzo Masoero\textsuperscript{2}, Stefano Favaro\textsuperscript{3,4}, and Thomas S. Richardson\textsuperscript{2,5}}

\begin{document}
\maketitle

\begingroup
\renewcommand{\thefootnote}{\arabic{footnote}}
\footnotetext[1]{Department of Economics, Management, and Statistics, University of Milano--Bicocca}
\footnotetext[2]{Amazon, Inc.}
\footnotetext[3]{Department of Economics and Statistics, Università di Torino}
\footnotetext[4]{Collegio Carlo Alberto}
\footnotetext[5]{Department of Statistics, University of Washington}
\endgroup
\setcounter{footnote}{5}

\begin{abstract}
Online A/B tests are the standard tool for data-driven decision-making at scale. 
Among the design choices with the largest impact on statistical power is the triggering mechanism: how many users to expose and for how long.
This often requires forecasting user engagement, i.e., whether enough users will trigger, and when a target participation level will be reached,  from limited pilot data. 
We introduce a Bayesian nonparametric model for predicting both new-user counts and total triggers, accommodating the heavy-tailed engagement patterns typical of web experiments. All predictive quantities can be computed without intensive numerical procedures such as Markov chain Monte Carlo (MCMC) or variational inference. We evaluate on three public datasets (over 450 public benchmark evaluations) and a proprietary benchmark drawn from 759 production A/B tests comprising 1,774 arms. Across the benchmark analyses, our models are competitive and frequently improve accuracy in forecasting new users, total triggers, and time to reach a target sample size compared with state-of-the-art competitors, especially when only a few pilot days are observed.

\end{abstract}

\section{Introduction}

In the era of data-driven decision making, online A/B testing has emerged as a fundamental tool for evaluating the effectiveness of new interventions on key performance indicators in the technology industry \citep{kohavi2013online,gupta2019top}. 
By randomly assigning users to different versions of a product or service via an A/B test, experimenters can measure the causal impact of their interventions and make informed decisions about product development and user experience. In this paper, however, we do not focus on causal-effect estimation itself; rather, we address the upstream design problem of forecasting participation and activity to plan experiment duration.
The success of an online A/B test relies heavily on user engagement, which reflects the users' interest in and interaction with the product.
Accurate prediction of user engagement, including the number of new users and their re-trigger counts, is crucial for planning and running efficient experiments.

Experimenters face a challenging trade-off when designing online A/B tests. On one hand, running shorter experiments can reduce costs and minimize the exposure of users to potentially suboptimal experiences. On the other hand, longer experiments can provide more representative data and improve the reliability of the results. 
Accurate predictions of user engagement can help experimenters navigate this trade-off and make informed decisions about the optimal experimental duration. 
By accurately forecasting the number of new users and their re-trigger counts, experimenters can ensure that they collect sufficient data to detect meaningful treatment effects while minimizing costs associated with prolonged experimentation.

Our primary goal in this paper is \emph{design and forecasting}: using a short pilot period to predict future participation and activity so that experimenters can decide how long to run an A/B test (or whether a test should be extended or stopped). Concretely, given a pilot window of length $D_0$, we focus on forecasting (i) the number of new users in a future window of length $D_1$, (ii) the total number of triggers in that window when re-trigger data are available, and (iii) the number of additional days $D_M$ needed to reach a total target of $M$ distinct users.
In practice, these outputs are used to choose an experiment duration or extension window: given a total target of $M$ distinct users, the method produces a duration-to-target recommendation that supports a concrete stop/extend decision under time or cost constraints. These predictive targets are intended to inform decisions about duration and resource allocation; estimating causal effects is a downstream task that benefits from improved design but is not the primary estimand here. Recent industry work emphasizes that experiment duration and power are inherently tied to prior beliefs about effect sizes and to principled early-termination rules \citep{Gualavisi2025,Masoero2025}; our framework complements these objectives by providing fast, accurate forecasts of participation and activity from short pilots.

Existing methods for predicting user engagement in online A/B tests often make restrictive assumptions about user behavior, limiting their ability to capture the complex dynamics of user engagement over time. For example, some methods assume that the propensity to engage is equal for all users (both in the control and treatment group) \citep{deng2015objective} or model user-specific propensities as i.i.d. random variables \citep{richardson22a}. Instead, in practice, one often observes a power-law behavior of users' activities, with few users being extremely active and most users seldom being active.
Other methods require specifying an artificial upper bound on the total number of users \citep{richardson22a, wan2023experimentation}, which can be challenging to estimate and may lead to biased predictions. 
Moreover, many existing methods focus solely on predicting the number of new users, neglecting the important information contained in re-trigger counts, which can provide valuable insights into user loyalty and long-term engagement.

These limitations motivate a forecasting framework with three features. 
First, it should allow for an unbounded population of potential users and for strong heterogeneity in user activity. 
Second, it should adapt to the level of information available in a given experiment, ranging from first-trigger times to binary daily activity and count-valued re-triggers. 
Third, it should produce predictive summaries that map directly to design decisions, such as choosing an extension window or estimating the additional time needed to reach a target number of distinct users.
To address these challenges, we propose a Bayesian nonparametric approach based on the stable beta-scaled process (SB-SP) prior \citep{camerlenghi2022scaled}.
The SB-SP prior provides a flexible model for heterogeneous user propensities without requiring a fixed upper bound on the number of users, while retaining analytical tractability.
We combine this prior with three observation models tailored to different data resolutions: a truncated geometric model for first-trigger data, a Bernoulli model for daily user activity, and a negative binomial model for re-trigger counts.
For each model, we derive posterior and predictive quantities that can be evaluated efficiently, avoiding computationally intensive posterior simulation in the large-scale forecasting tasks considered here.

The resulting framework is intended to connect statistical forecasting with experiment planning rather than to serve as a stand-alone descriptive model of engagement.
Through extensive simulations and real-world applications, we show that the proposed SB-SP-based methods achieve competitive or improved predictive performance for new-user and activity forecasting, and that these gains translate into more accurate duration-to-target decisions.
The main contributions are:
\begin{enumerate}
    \item We introduce a Bayesian nonparametric framework for \emph{forecasting user engagement to support experiment design} in online A/B tests, including forecasts of new users, total future triggers, and duration-to-target quantities.
    \item We develop SB-SP-based models for first-trigger, binary-activity, and count-valued engagement data, deriving closed-form posterior and predictive expressions that enable scalable inference from short pilot periods.
    \item We demonstrate the predictive accuracy and practical utility of the approach through simulations and real-data analyses, showing how accurate forecasts support run-length and extension decisions in A/B testing.
\end{enumerate}

The rest of the paper is organized as follows.
\Cref{sec:setting} describes the data and challenges of predicting user engagement in online A/B tests and reviews existing methods.
\Cref{sec:bnp_method} presents our Bayesian nonparametric methodology, including the stable beta-scaled process prior, likelihood models, and inference procedures. 
\Cref{sec:numerical_implementation} discusses a practical numerical implementation of the derived methods.
We report the experimental results from simulation studies in \Cref{sec:simulations} and on real-world data in \Cref{sec:exp_real}, comparing our approach to established alternatives. Finally, \Cref{sec:discussion} discusses the implications of our findings, potential applications, and future research directions.

\section{Motivating application, data, and existing methods} \label{sec:setting}

% Section 2.1: User engagement data in online A/B tests.

\subsection{User Engagement Data in Online A/B Tests}

In online A/B testing, users are randomly assigned to either the control or the treatment group.
Each of these groups encodes a different version of the product or service that the experimenter is interested in testing. 
The goal is to measure the causal impact of the treatment on user behavior and key performance indicators \citep{kohavi2013online}. 
Within these tests, user engagement is a crucial component, as it reflects the users' interest in and interaction with the product.

We consider data from online A/B tests consisting of daily user activity over a fixed time period. For each user $\useridx$, we observe the day 
of their first trigger in the experiment $F_\useridx \in \NN$ (i.e., the time and day they first engage with the product) and, possibly, the count $\daycount \geq 0$ of 
% their subsequent re-triggers for each day $d$. 
%% TSR: This is a bit imprecise since the first
%% trigger is included.
the total number of triggers for each day $d$.
In certain settings, where the service running the experiment only records coarse activity data, $\daycount$ might simply be a binary indicator of whether user $\useridx$ engaged on day $\dayidx$. 
The data can be represented as a set of tuples $\{(\omega_\useridx, F_\useridx, A_{1:\daytotal,\useridx})\}_{\useridx=1}^\usertotal$, where $\omega_\useridx$ is a unique identifier for user $\useridx$, $\daytotal$ is the total number of days in the experiment, and $\usertotal$ is the total number of users who engaged with the product during the experiment. Note that $F_\useridx$ is determined by $A_{1:D,\useridx}$ via $F_\useridx = \min\{d : A_{d,\useridx} > 0\}$; we include it explicitly for notational convenience.
We here use the shorthand notation $X_{a:b} := X_a, X_{a+1}, \ldots, X_b$.

One of the main challenges in modeling user engagement data is its sparsity and heterogeneity. 
In typical online A/B tests, a large proportion of users may engage with the product only once or a few times, while a small fraction of users may exhibit high levels of engagement. 
Moreover, users may have different propensities to engage with the product over time, leading to complex dynamics in the re-trigger counts. 
Capturing this heterogeneity and temporal dependence is crucial for accurately predicting user engagement.

% Section 2.2: Data description.

\subsection{Data description}

In this paper, we consider four datasets: three public datasets (REES46, UCI, and ASOS) spanning different scales and domains, as well as a fourth proprietary dataset.
For the three public datasets, we partition the observation period into experiment windows consisting of a short pilot period (typically 7 days) followed by a longer follow-up period, see Appendix~\ref{app:benchmark_construction} for the full specification.
Since our targets---new-user counts, total triggers, and time-to-threshold---are forecasting targets rather than treatment-effect estimands, 
observational transaction logs serve as valid prediction benchmarks whenever ground-truth future participation is known.
These benchmarks do not assume that, in a randomized experiment, triggering behavior is invariant to treatment assignment.
UCI and REES46 are such observational logs; ASOS provides data from actual randomized experiments.
\begin{itemize}
    \item \textbf{REES46 eCommerce} \citep{rees46dataset}. The dataset contains approximately 380M event records from 15.6M users over an interval spanning 213 calendar days (Oct 2019--Apr 2020), with events observed on 211 dates.
    Per-event timestamps allow construction of per-user daily trigger counts.
    We extract 7 non-overlapping 28-day windows (7-day pilot, 21-day follow-up).
    For robustness, we additionally construct rolling-window benchmarks with follow-up lengths of 21, 50, and 100 days. The pilot sizes (i.e., number of users active in the 7-day pilot) range between 957k and 1.8M users.
    \item \textbf{UCI Online Retail} \citep{chen2012data}. The dataset contains observational logs from
    397K transactions from a total of 4,339 customers over a period of 374 days (Dec 2010--Dec 2011).
    We construct data by partitioning the logs into 13 non-overlapping 28-day windows (7-day pilot, 21-day follow-up). The pilot sizes range between 34 and 474 users.
    \item \textbf{ASOS} \citep{liu2021datasets}. The dataset consists of first-trigger counts from
    72 fashion experiments (144 arms).
    Follow-up periods range from 7 to 239 days, with a median of 55.
    \item \textbf{Proprietary data}. The collection contains 759 production A/B tests comprising 1,774 arms. For the new-user accuracy benchmark, we analyze one eligible control arm per test using a 7-day pilot and a 21-day follow-up. The pilot sizes range from a few hundred to hundreds of millions.
\end{itemize}

We consider the REES46 dataset to illustrate the empirical patterns that motivate the
modeling choices in \Cref{sec:bnp_method}, but similar patterns can be observed in all the datasets. As shown in \Cref{fig:powerlaw}, the per-user trigger-count
distribution is approximately power-law \citep{clauset2009power}: 24\% of users
trigger exactly once, the median user triggers 5 times, and the maximum exceeds
199K triggers. Participation is also sparse: 49\% of users are active on exactly one day, and 74\% are active on at most three of the 211 dates with recorded events; see \Cref{fig:days_active}.
Finally, daily new-user arrivals decay over time: day-1 arrivals in REES46
average 190K, whereas by day 28 they drop to about 60K. This concave trajectory
is incompatible with constant-rate models; see \Cref{fig:cumulative}. These patterns---power-law
heterogeneity, extreme sparsity, and decelerating arrivals---are precisely the
kind of behavior induced by the stable beta-scaled process prior.

\begin{figure}[t]
    \centering
    \begin{subfigure}[t]{0.48\linewidth}
        \includegraphics[width=\linewidth]{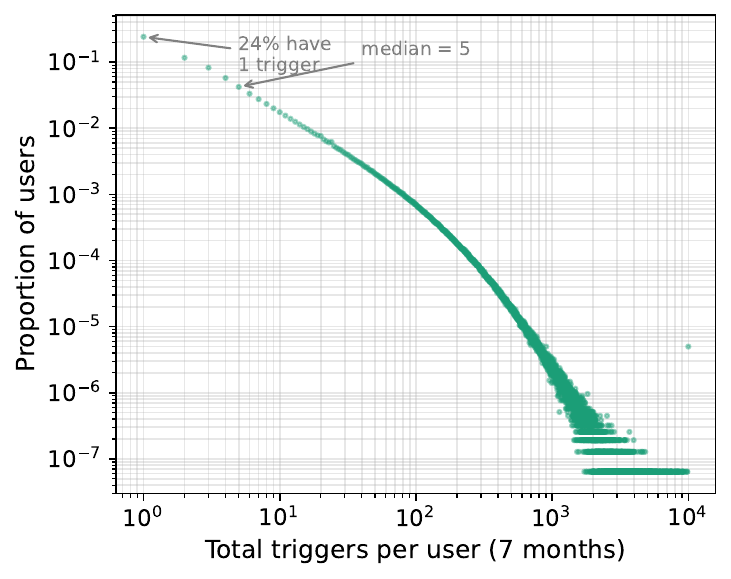}
        \caption{Per-user trigger counts (log-log). 24\% have exactly 1 trigger; median is 5.}
        \label{fig:powerlaw}
    \end{subfigure}
    \hfill
    \begin{subfigure}[t]{0.48\linewidth}
        \includegraphics[width=\linewidth]{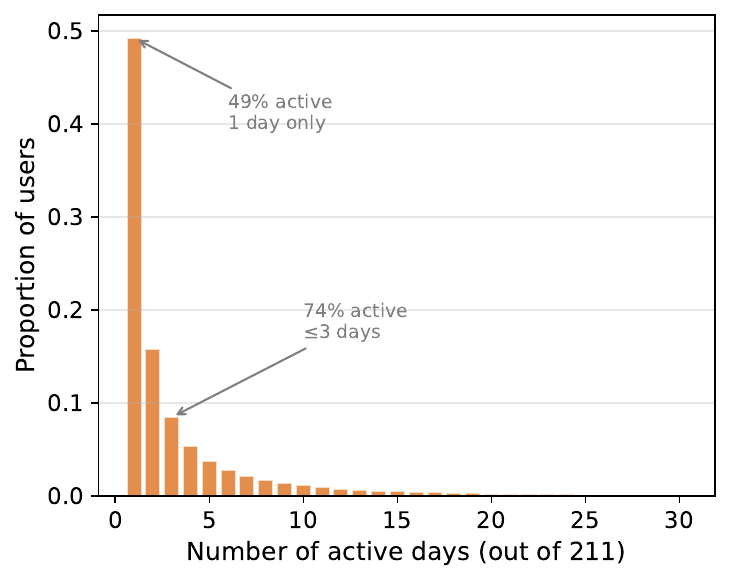}
        \caption{Days active per user. 49\% active on 1 day; 74\% on $\leq 3$.}
        \label{fig:days_active}
    \end{subfigure}
    \caption{REES46 user engagement distributions (15.6M users, 7 months).}
    \label{fig:engagement_distributions}
\end{figure}

\begin{figure}[t]
    \centering
    \includegraphics[width=0.65\linewidth]{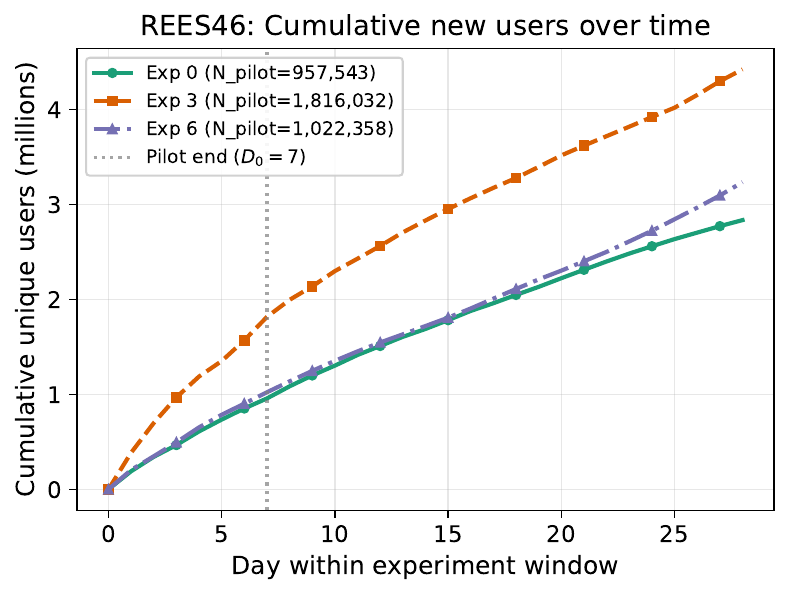}
    \caption{Cumulative new users over time for three 28-day experiment windows (REES46). Dashed line marks pilot end ($\pilotdays = 7$). All curves show sub-linear growth.}
    \label{fig:cumulative}
\end{figure}

\subsection{Existing methods for predicting user engagement} \label{sec:prev_work}

Several methods have been proposed to predict user engagement in online A/B tests and related settings.
Our primary focus is on forecasting quantities that inform experiment design and duration decisions, in line with recent industry work on power calculations under effect-size priors and Bayesian early-termination rules \citep{Gualavisi2025,Masoero2025}.
We here consider the problem of predicting the number of users that did not trigger in the first $D_0$ days (the \emph{pilot} study) but will trigger in days $D_0 +1, \ldots, D_0 + D_1$ (the \emph{follow-up} period), and cast it as an \emph{unseen feature problem}.
Formally, the goal of the analysis is to estimate 
\begin{equation}
    \news{D_0}{D_1} := 
    \sum_{\useridx \geq 1} 
        I
            \left\{
                \Bigg[
                    \sum_{\dayidx=1}^{D_0} \daycount  = 0 
                \Bigg] 
                \cap  
                \Bigg[
                    \sum_{\dayidx=D_0 + 1}^{D_0 + D_1} \daycount  > 0 
                \Big] 
            \right\}, \label{eq:news}
\end{equation}
where the outer sum ranges over all possible users.

We can broadly categorize methods to estimate $\news{D_0}{D_1}$ into statistical and algorithmic approaches.
In recent years, many interesting approaches have been developed in the context of genomic studies. 
In the former category, \citet{zou2016quantifying} proposed to solve the unseen feature problem via a linear programming approach and \citet{gravel2014predicting} employed jackknife estimators based on the first few values of the resampling frequency spectrum. More recently, \citet{chakraborty2019using} proposed to use the celebrated Good--Toulmin estimator \citep{good1956number} to predict the number of new genomic variants based on the observed frequency counts. Closely related is a classical line of work in the database and survey-sampling literature on estimating the number of distinct values in a population from a sample \citep{shlosser1981estimation, haas1995sampling, haasstokes1998estimating}.
In the latter category, \citet{richardson22a} focused on user engagement prediction and proposed a Bayesian model for the first trigger times. See \citet{ionita2009estimating, masoero2022more, camerlenghi2022scaled} for
other applications.

These existing approaches have limitations in online experiments with very short pilots.
Algorithmic approaches and frequency-spectrum extrapolators such as jackknife and Good--Toulmin are sensitive to regularization/tuning and can deteriorate markedly when extrapolating far beyond a short pilot window, often providing only narrow predictive summaries. 
Within the statistical (and in particular Bayesian) methods, \citet{ionita2009estimating, richardson22a} assume that the total number of users, say $M$, is fixed and known. 
Beyond the obvious difficulty of picking such a parameter from the available data, this assumption, together with the common assumption that each user has a nonzero probability of engaging any given day, entails that $\news{D_0}{D_1} \rightarrow M - \usertotal_{D_0}$ almost surely as $D_1 \rightarrow \infty$, where $\usertotal_{D_0}$ is the number of users that triggered in the first $D_0$ days.
In contrast, user activity often exhibits a power law growth, i.e., $\news{D_0}{D_1} \sim D_1^{\gamma}$ for some $\gamma \in (0, 1)$.
\citet{masoero2022more, camerlenghi2022scaled} assume two different Bayesian nonparametric models, both based on a generalization of the celebrated Indian buffet process \citep{Gri(11)}. 
In their framework, the total number of users is unbounded, but the number of users that trigger in any finite time-frame is finite, almost surely, and the number of active users exhibits a power law growth.
A limitation of both models is that they can only be applied to coarse activity data, i.e., when $\daycount$ is a binary indicator of whether user $\useridx$ is active on day $\dayidx$.
This implies losing information when using granular daily activity data.
Moreover, these models assume that each user's propensity to engage remains constant in time, which can be an oversimplification in several settings and bias the inference. In our experience, in such a setting, modeling only the first trigger times $F_\useridx$ might be preferable, as we demonstrate in our experiments.

Motivated by these considerations, this paper focuses on statistical approaches to the unseen feature problem and its generalizations. 
In particular, we take a Bayesian model-based approach for two reasons. 
First, our data consists of only a handful of experiment days (namely $D_0 \leq 7$).
Bayesian methods are well-suited for settings where predictions have to be formed with limited amounts of input ``training'' data, thanks to the possibility of eliciting informative prior information that regularizes the inference.
Moreover, our goal is not user prediction per se but rather to inform online experiments that must decide the duration of A/B tests and whether a specific experiment should be terminated immediately after observing the first $D_0$ days.
Therefore, fast forecasts of future participation and duration-to-target quantities help experimenters make more informed run-length decisions.
From the technical standpoint, our results are based on a generalization of the mathematical framework in \cite{camerlenghi2022scaled}. In particular, we adopt the same stable beta-scaled process prior, but allow for a greater flexibility in the choice of the likelihood: while the model in \cite{camerlenghi2022scaled} is suitable for modelling daily activity of each user, our results allow for modelling also the first trigger times of users, as well as the total number of daily clicks. Beyond the likelihood extension, our inferential targets are explicitly decision-oriented (multi-horizon forecasts, total-trigger prediction, and time-to-threshold $D_M$), which are not addressed in \cite{camerlenghi2022scaled} though they are central to A/B test duration and early-stopping decisions.

\section{Bayesian Nonparametric Methodology}\label{sec:bnp_method}

This section presents our Bayesian approach for predicting user engagement in online A/B tests. 
We first set up a general framework encompassing all previously proposed Bayesian models as special cases.
We highlight how the choice of the prior distribution plays a crucial role for the quality of the posterior and predictive distributions of interest.
We then discuss using a stable beta-scaled process prior and show how this results in a flexible model that retains analytical tractability and closed-form expressions for posterior inferences.
We frame our model within the class of Bayesian trait processes \citep{Jam(17),Cam(18), masoero2018posterior, Ber(23)}.  

\subsection{Trait process formulation}

First, we show how the data structures discussed in \Cref{sec:setting} can all be cast under the same umbrella of a \emph{trait process}.
To this end, recall that $\omega_n$ is the unique user-id for the $n$-th user, $F_n$ is their first trigger time, and $A_{d, n}$ is their activity for the $d$-th day.
Then, we can represent the data collected during day $d$ of the experiments as a counting measure supported on the user ids:
\begin{equation}\label{eq:trp1}
    Z_\dayidx(\cdot) = 
        \sum_{n\ge 1} A_{d,n} \delta_{\omega_\useridx}(\cdot), \quad \dayidx=1, \ldots, D.
\end{equation}
From the definition above, the relation $Z_\dayidx(\omega_\useridx) = \daycount$ gives a way of pairing each daily activity index $\daycount$ of the $\useridx$-th unit on the $d$-th day together with the label $\omega_\useridx$ of the $\useridx$-th unit.
Similarly, let
\begin{equation}\label{eq:trp2}
    Z^{\FT}(\cdot) = \sum_{\useridx \ge 1} F_\useridx \delta_{\omega_\useridx}(\cdot),
\end{equation}
which is a measure that collects only the first triggering time (FT) of all active users. 
Here, we let both measures $Z_\dayidx$ and $Z^{\FT}$ be supported over a countably infinite set
by letting $\daycount = 0$ for any user $\useridx$ that did not trigger on day $\dayidx$ and $F_\useridx = 0$ for all users that did not trigger in the first $D$ days.
The value $F_\useridx = 0$ is a censoring marker for ``no trigger within the first $D$ days''. The first-trigger score distribution introduced below in \eqref{eq:tr_geom} places mass $(1-\theta)^{D}$ on the event $\{F_\useridx = 0\}$. Predictions for a subsequent interval concern users who were inactive during the pilot and are governed by the posterior of the latent measure $\mu$ restricted to those users.

Trait processes \citep{Cam(18), Jam(17)} provide a convenient mathematical formalism to assign Bayesian models to data as in \eqref{eq:trp1} and \eqref{eq:trp2}. We recall the general definition of a trait process below.
\begin{definition}[Trait Process]\label{def:trp}
    Let $\mu = \sum_{\useridx \geq 1} \theta_\useridx \delta_{\omega_\useridx}$ (with scalar weights $\theta_n \ge 0$) be an almost surely discrete random measure with distribution $\mathscr P_\mu$, where $\delta_x$ is the Dirac delta at $x$. Given $\mu$, let $Z_\dayidx ~=~ \sum_{\useridx \geq 1} \daycount \delta_{\omega_\useridx}$, $\dayidx \geq 1$ such that for a probability mass function $\mathcal G$ with support on $\NN:=\{0,1,2,\ldots\}$ (the \emph{score distribution}, i.e., the conditional law of $A_{d,n}\mid\theta_n$),
    \[
        \daycount \mid \mu \ind \mathcal G(\cdot \mid \theta_\useridx).
    \]
    Then we say that $(Z_\dayidx)_{\dayidx \geq 1}$ is a trait process directed by $\mu$ with score distribution $\mathcal G$. We write
    \begin{equation}\label{eq:trait_model}
          Z_{\dayidx} 
          \mid \mu \iid 
          \mathrm{TrP}(\mathcal G; \mu),
          \quad \mu \sim \mathscr P_\mu.
    \end{equation}
\end{definition}
Coming back to our data, we assume that $Z_\dayidx$'s in \Cref{eq:trp1} follow a trait process with score distribution defined as follows. 
Throughout, we use $\pi_{\mathcal G}(\theta) := 1-\mathcal G(0\mid\theta)$, so that $\pi_{\mathcal G}(\theta)$ is the probability of positive activity and $1-\pi_{\mathcal G}(\theta)=\mathcal G(0\mid\theta)$ is the probability of zero activity. We use $\mathrm{Gamma}(a,b)$ for the gamma distribution with shape $a$ and rate $b$.
When the data consists of daily indicator variables for each user, then $\mathcal G = \mathcal G_{\mathrm{Be}}$ is the Bernoulli distribution
\[
    \mathcal G_{\mathrm{Be}}(a\mid\theta)=\theta^a(1-\theta)^{1-a},\qquad a\in\{0,1\},
\]
and $\pi_{\mathrm{Be}}(\theta)=\theta$.
If, instead, we have daily re-trigger data for each user ($\daycount \in \NN$), we let $\mathcal G = \mathcal G_{\mathrm{NB}}$ be the negative binomial score distribution with fixed, user-independent number of failures $r>0$
\[
    \mathcal G_{\mathrm{NB}}(a\mid\theta; r)=\binom{a+r-1}{a}(1-\theta)^r\theta^a,\qquad a=0,1,2,\ldots.
\]
Hence $\mathcal G_{\mathrm{NB}}(0\mid\theta; r)=(1-\theta)^r$ and $\pi_{\mathrm{NB}}(\theta)=1-(1-\theta)^r$.
The commonly employed prior distribution $\mathscr P_\mu$ are such that $\theta_n \le \varepsilon$ infinitely often for any choice of threshold $\varepsilon$. Therefore, letting the negative binomial distribution success probability be $(1- \theta)$ entails that $A_{d, n} = 0$ infinitely many times and, more importantly, $A_{d,n} > 0$ only finitely many times. That is, under our model, for any time horizon $D$, only finitely many users will be active, which is what one would expect. Letting the success probability be $\theta$ instead would produce the opposite, with infinitely many users being active at any given time (under the model), which is not a reasonable modeling assumption. See \cite{Ber(23)} for a proof.
Similarly, we assume a trait process for \eqref{eq:trp2} where $\mathcal G = \mathcal G_{\mathrm{TG}}$ is the probability mass function for a truncated geometric distribution over $\{0, 1, \ldots, \daytotal\}$ given by 
\begin{equation}\label{eq:tr_geom}
    \mathcal G_{\mathrm{TG}}(y\mid \theta; D) = \begin{cases}
        (1 - \theta)^{y-1} \theta & \text{ if }  
        y \in \{1,\ldots , D\},\\
%        1 \leq y \leq D, \\
        (1 - \theta)^D & \text{ if } y = 0,
    \end{cases}
\end{equation}
i.e.\ $A \sim \mathcal G_{\mathrm{TG}}(\cdot\mid\theta;D)$ is equivalent to $A^\prime \sim \mathrm{Geom}(\theta)$ and $A = A^\prime$ if $1 \le A^\prime \le D$ and $A=0$ otherwise. Thus $\pi_{\mathrm{TG}}(\theta;D)=1-(1-\theta)^D$.

All the previously proposed Bayesian methodologies for the unseen feature problem (cf.\  \Cref{sec:prev_work}) can be recovered as special cases of this trait process framework. 
In particular, the model in \citet{masoero2022more} is equivalent to \Cref{eq:trp1} where $\mathcal G = \mathcal G_{\mathrm{Be}}$ and $\mathscr P_\mu$ is the law of the (three-parameter) beta process \citep{Teh(09), Bro(12)}.
Similarly, the beta-binomial model of \citet{ionita2009estimating} is recovered from \eqref{eq:trp1} when $\mathcal G = \mathcal G_{\mathrm{Be}}$ and $\mu = \sum_{\useridx = 1}^{\usertotal} \theta_\useridx \delta_{\omega_\useridx}$ where $\usertotal$ is the total number of users (assumed fixed) and $\theta_\useridx \iid \mathrm{Beta}(a, b)$.
Finally, the hierarchical beta-geometric model of \citet{richardson22a} is recovered from \Cref{eq:trp2} choosing $\mathcal G_{\mathrm{TG}}$ and assuming the same prior for $\mu$ as in the beta-binomial model.
Beyond the likelihood distribution, one key aspect of these Bayesian models is the prior. From the general theory on trait processes in \cite{Jam(17)}, \cite{Cam(18)} and \cite{Ber(23)}, it is evident that the prior plays a crucial role in prediction. 
While previously proposed Bayesian methods make restrictive assumptions about the form of $\mu$, limiting their ability to model sparse and diverse user engagement patterns, we argue in favor of the stable beta-scaled process (SB-SP) prior \citep{camerlenghi2022scaled}, which allows for power-law behaviors (both a priori and a posteriori), and a rich predictive structure \citep{ExtFeatures}.

% Hence, while the likelihood models described above can accommodate various types of user engagement data, the choice of the prior distribution for $\mu$ is crucial for capturing the complex dynamics and heterogeneity of user behavior. 
% Previously proposed Bayesian methods make restrictive assumptions about the form of $\mu$, limiting their ability to model sparse and diverse user engagement patterns.
% To address these limitations, we propose using the stable beta-scaled process (SB-SP) prior \citep{camerlenghi2022scaled} for $\mu$.

\subsection{Stable Beta-Scaled Process Prior}

The SB-SP prior is a flexible and tractable prior distribution that allows for a rich representation of user heterogeneity while maintaining desirable properties such as exchangeability and analytical tractability. 
To define it, consider an $\alpha$-stable random measure $\mu = \sum_{\useridx \geq 1} \tau_\useridx \delta_{\omega_\useridx}$. That is, $\mu$ is a completely random measure (CRM) with L\'evy intensity $\alpha s^{-1 - \alpha} \dd s P_0(\dd x)$ for $0 < \alpha < 1$ 
and $P_0$ a diffuse probability measure.
Denote by $\Delta_1 > \Delta_2 > \cdots$ the decreasingly ordered random jumps $\tau_{\useridx}$ of $\mu$.
Following \cite{Fer(72)} $\Delta_1$ has density $f_{\Delta_1}(\zeta) = \exp\{-\alpha \int_\zeta^{+\infty} s^{-1-\alpha} \dd s\} \alpha \zeta^{-1-\alpha}=\exp\{-\zeta^{-\alpha}\}\alpha \zeta^{-1-\alpha}$, i.e., $\Delta_{1}^{-\alpha} \sim \mathcal{E}(1)$, where $\mathcal{E}(1)$ denotes an exponential distribution with rate $1$.
Denote by $\mathcal L_\zeta(\cdot)$ the conditional distribution of $(\Delta_{\useridx+1} / \Delta_1)_{\useridx \geq 1}$ given $\Delta_{1} = \zeta$. 
Then, a scaled process is obtained by marginalizing $\Delta_1$ from the latter distribution. 
As noted in \cite{james15sp}, we can gain in flexibility by changing the law of $\Delta_1$, i.e., marginalizing $\mathcal L_\zeta(\cdot)$ with respect to $\zeta \sim h^*$ for any distribution $h^*$ supported on the non-negative reals.
The stable beta-scaled process is obtained by a suitable choice of $h^*$.
\begin{definition}\label{def:sbsp}
    A stable beta-scaled process (SB-SP) prior is the law of random measure $\mu = \sum_{\useridx \geq 1} \theta_\useridx \delta_{\omega_\useridx}$ 
    where $\omega_\useridx \iid P_0$, a diffuse probability measure,
    and $(\theta_\useridx)_{\useridx \geq 1}$ 
    is distributed as $\int \mathcal L_{\zeta}(\cdot) h_{\alpha, c, \beta}(\zeta) \dd \zeta$ where
    \[
         h_{\alpha, c, \beta}(\zeta) = \Gamma(c+1)^{-1} \alpha \beta^{c+1} \zeta^{- \alpha (c+1) - 1} \exp \{-\beta  \zeta^{-\alpha}\},
    \]
    for $0 < \alpha < 1$, $\beta > 0$, and $c > 0$.
    We will use the notation $\mu \sim \mbox{SB-SP}(\alpha, c, \beta)$.
    We denote the corresponding latent scaling variable by $\Delta_{1,h}$, so that $\Delta_{1,h}$ has density $h_{\alpha,c,\beta}$; equivalently, $\Delta_{1,h}^{-\alpha}\sim\mathrm{Gamma}(c+1,\beta)$.
\end{definition}
We introduce the following notation for shorthand convenience.
\begin{definition}  \label{def:models}
    Let $\mu \sim \text{\rm SB-SP}(\alpha, c, \beta)$. 
    If $Z_\dayidx$ is as in \Cref{eq:trp1},  $Z_\dayidx ~\mid \mu \ind~ \mathrm{TrP}(\mathcal G_{\mathrm{NB}}; \mu)$, we say that $Z_{1:D} = (Z_1, \ldots, Z_\daytotal)$ follows a negative binomial stable beta scaled process and write $Z_1, \ldots, Z_\daytotal \sim \mbox{\nbmodel}$.
    If, instead, $\mathcal G = \mathcal G_{\mathrm{Be}}$, we say that $Z_{1:D}$ follows a Bernoulli stable beta scaled process and write $Z_1, \ldots, Z_\daytotal \sim \mbox{\bemodel}$.
    Finally, if $Z^{\FT}$ is as in \eqref{eq:trp2}, such that $Z^{\FT} \mid \mu ~\sim~ \mbox{TrP}(\mathcal G_{\mathrm{TG}}; \mu)$ where $\mathcal G_{\mathrm{TG}}$ is as in \eqref{eq:tr_geom}, we say that $Z^{\FT}$ follows a truncated geometric stable beta scaled process and write $Z^{\FT} \sim \text{\tgmodel}$.
\end{definition}

\subsection{Posterior Distribution and Predictive Quantities}

Central to obtaining computationally feasible estimators for the future user engagement is the study of the posterior and predictive distributions under the trait process model. 
By extending the Poisson partition calculus developed in \cite{Jam(17)} for CRM priors, the next Theorem gives a unified template for the posterior distribution under the SB-SP priors.

We introduce some additional notation.  For a sample $Z_{1:\daytotal}$ from \eqref{eq:trait_model}, let $\pi_{\mathcal G}(\theta)=1-\mathcal G(0\mid\theta)$ denote the probability of positive activity under score distribution $\mathcal G$, and let $\omega^*_{1:\usertotal_\daytotal} = \{ \omega_n : A_{d,n} >0 \text{ for some } d \in \{1,\ldots , D\} \}$ be the $\usertotal_\daytotal$ unique feature labels displayed in such a sample  (i.e., in our notation, the unique user identifiers of those active users). For $\useridx = 1, \ldots, \usertotal_\daytotal$ define
\[
    B_\useridx = \{d\le D: Z_d(\omega^*_\useridx)>0\},\qquad
    b_\useridx = |B_\useridx|,\qquad
    t_\useridx = \sum_{d=1}^D Z_d(\omega^*_\useridx).
\]
That is, $b_\useridx$ is the number of active days and $t_\useridx$ is the total trigger count. Observe that $B_n$, $b_n$, $t_n$, and $N_D$ are random.

\begin{theorem}\label{thm:post_general}
    Let $Z_{1:\daytotal}$ be a sample from \eqref{eq:trait_model}, where $\mu \sim \text{\rm SB-SP}(\alpha, c, \beta)$. Let $\tilde\Delta_{1,h}$ have posterior law $\mathcal L(\Delta_{1,h}\mid Z_{1:\daytotal})$.
Then, the posterior distribution of $\mu$ given $Z_{1:\daytotal}$ coincides with the law of 
\[
    \sum_{\useridx = 1}^{\usertotal_\daytotal} \theta^*_\useridx \delta_{\omega^*_\useridx} + \mu^\prime,
\]
such that  
\begin{enumerate}
    \item $\theta^*_1,\ldots,\theta^*_{\usertotal_\daytotal}$ are independent random variables with marginal density
    \begin{equation}\label{eq:post_jump_general}
        f_{\theta^*_\useridx}(\theta) \propto (1-\pi_{\mathcal G}(\theta))^{\daytotal-b_\useridx}\left\{\prod_{\dayidx \in B_\useridx} \mathcal G(Z_\dayidx(\omega^*_\useridx)\mid \theta)\right\}\theta^{-1-\alpha}\mathbf 1_{(0,1)}(\theta).
    \end{equation}
    \item $\mu^\prime$ is a random measure independent of the $\theta^*_\useridx$'s with law
    \begin{equation}\label{eq:post_mu_general}
    \begin{aligned}
        \mu^\prime \mid \tilde \Delta_{1,h}
        &\sim \mathrm{CRM}(\nu'_{\tilde \Delta_{1,h}}),\\
        \nu'_{\tilde \Delta_{1,h}}(\dd\theta,\dd\omega)
        &=
        \tilde \Delta_{1,h}^{-\alpha}\alpha
        (1-\pi_{\mathcal G}(\theta))^{\daytotal}
        \theta^{-1-\alpha}
        \mathbf 1_{(0,1)}(\theta)
        \dd\theta P_0(\dd\omega),\\
        \tilde \Delta_{1,h}^{-\alpha} \mid Z_{1:\daytotal}
        &\sim \mathrm{Gamma}\left(
            \usertotal_\daytotal+c+1,
            \beta+\alpha\int_0^1
            \{1-(1-\pi_{\mathcal G}(\theta))^\daytotal\}
            \theta^{-1-\alpha}\dd\theta
        \right).
    \end{aligned}
    \end{equation}
\end{enumerate}
\end{theorem}

\Cref{thm:post_general} decomposes the posterior distribution of the random measure $\mu$ as the sum of two terms: the first one corresponding to the already observed feature labels (i.e., the user ids of those users that already triggered) and the second one corresponding to the hitherto unobserved feature labels (i.e., users who are, thus far, inactive). According to the theorem, each active user $\omega^*_n$ is associated with parameters $\theta^*_n$ that depend exclusively on their activity in the sample, independently across users. On the other hand, the potentially infinite unobserved users are collected in a random measure $\mu^\prime$ whose distribution corresponds to a tilting of the original SB-SP prior law. Moreover, \Cref{thm:post_general} allows us to obtain a formal description of the predictive distribution for $Z_{\daytotal+1} \mid Z_{1}, \ldots, Z_\daytotal$. 
Indeed, we have
\begin{equation}\label{eq:pred_formal}
    Z_{\daytotal+1} \mid Z_{1}, \ldots, Z_\daytotal 
    = 
    \sum_{\useridx = 1}^{\usertotal_\daytotal} {Z}^*_{\daytotal+1, \useridx} \delta_{\omega^*_\useridx} + Z^\prime_{\daytotal+1},
\end{equation}
where $Z^*_{\daytotal+1, \useridx} \mid \theta^*_\useridx \sim \mathcal G(\cdot; \theta^*_\useridx)$ and $Z^\prime \mid \mu^\prime \sim \mathrm{TrP}(\mathcal G; \mu^\prime)$ for $\theta^*_\useridx$ and $\mu^\prime$ as in Theorem \ref{thm:post_general}.

To specialize \Cref{thm:post_general} to the three models under consideration, we observe the following. If data is as in \eqref{eq:trp1}, the sample size $\daytotal$ corresponds to the number of days $\daytotal_0$; instead, if data is as in \eqref{eq:trp2}, the sample size is equal to one. In both cases, the number of unique features $\usertotal_\daytotal$ in the sample is the number of users that were active in the pilot study (i.e., the first $\daytotal_0$ days), namely $\usertotal_{\daytotal_0}$.
Then, we obtain the following expressions for the posterior distributions under the three models considered.
\begin{corollary}\label{cor:post_models}
    Let $\mathrm{B}(a,b) = \Gamma(a)\Gamma(b)/\Gamma(a+b)$ be the beta function. For $x\ge 0$, $y\ge 0$, and $r>0$, let
    \[
        \psi_r(x,y) := \alpha\int_0^1 (1-\theta)^{rx}\{1-(1-\theta)^{ry}\}\theta^{-1-\alpha}\dd\theta.
    \]
    Equivalently, $\psi_r(x,y)=\alpha\{\mathrm{B}(rx+1,-\alpha)-\mathrm{B}(r(x+y)+1,-\alpha)\}$, where the beta-function expression is understood by analytic continuation. In particular, $\psi_r(x,0)=0$. Let $r_*=r$ for \nbmodel and $r_*=1$ for \bemodel and \tgmodel.
    Then for all the BNP models considered --- \nbmodel, \bemodel, or \tgmodel --- the posterior of $\mu$ is equivalent to the law of
    \begin{equation}\label{eq:post_mu_activity}
        \sum_{\useridx = 1}^{N_{D_0}} \theta^*_\useridx \delta_{\omega^*_\useridx} + \mu^\prime, \quad \mu^\prime = \sum_{\ell \geq 1} \theta^\prime_\ell \delta_{\omega^\prime_\ell}.
    \end{equation}
    In particular, $\mu^\prime$ collects all the user-specific parameters $\theta^\prime_\ell$ of those users $\omega^\prime_\ell$ that did not trigger in the first $D_0$ days. 
    Moreover,
    \[
        \tilde \Delta_{1,h}^{-\alpha}\mid data \sim \mathrm{Gamma}\left(N_{D_0}+c+1,\beta+\psi_{r_*}(0,D_0)\right).
    \]
    For observed users, $\theta^*_\useridx\mid data\sim\mathrm{Beta}(a_\useridx,b_\useridx^{post})$ with
    \[
        (a_\useridx, b_\useridx^{post}) = \begin{cases}
            (t_\useridx - \alpha, r D_0 + 1),\quad
            t_\useridx=\sum_{d=1}^{D_0} A_{d,\useridx}
            & \text{under the \nbmodel model}, \\
            (b_\useridx - \alpha, D_0 - b_\useridx + 1),\quad
            b_\useridx=\sum_{d=1}^{D_0} A_{d,\useridx}
            & \text{under the \bemodel model}, \\
            (1 - \alpha, F_\useridx)
            & \text{under the \tgmodel model}.
        \end{cases}
    \]
\end{corollary}
% We remark the deep similarities between the prosterior laws under the three models. In particular, and 
Observe that, under the \bemodel and the \nbmodel\!\!, the posterior of $\theta^*_n$ depends on the active-day count $b_n$ and the total trigger count $t_n$, respectively. In particular, if the relevant count is large, $\theta^*_n$ increases, meaning that we expect to see that user often also in the future. Instead, under the \tgmodel\!\!, $\theta^*_n$ decreases with $F_n$, which is also to be expected: if a user is active in the first stages of the experiment, they will likely have a high trigger rate.

Rather surprisingly, the posterior for $\mu^\prime$ coincides for the $\mbox{\bemodel}$ and $\mbox{\tgmodel}$ models. 
This entails that the predictive distribution of quantities that can be expressed as functionals of $\mu^\prime$ will be the same under $\mbox{\bemodel}$ and $\mbox{\tgmodel}$ models because $(1-\pi_{\mathrm{Be}}(\theta))^{D_0}=(1-\theta)^{D_0}$ and $1-\pi_{\mathrm{TG}}(\theta;D_0)=(1-\theta)^{D_0}$. For the \tgmodel\!\!, $\mathcal G_{\mathrm{TG}}(\cdot\mid\theta;D_0)$ is one trait-process observation whose likelihood already depends on the pilot horizon $D_0$.
This means that for predictive purposes, the \tgmodel and the \bemodel appear indistinguishable. However, one key difference is the marginal distribution they induce over the observations, which depends on different sufficient statistics in the two models, as we will show in Theorem \ref{thm:marg_general} below. Hence, even if the predictive rules under the two models are identical for fixed parameter values, they can differ significantly once one adopts an empirical Bayesian procedure to estimate the parameters from the marginal likelihood of the data.
From an applied perspective, this tells practitioners that the choice between \bemodel and \tgmodel is not about the closed-form predictor itself (which matches for functionals of $\mu^\prime$), but about which data summary is most defensible for the application: daily activity counts versus first-trigger times. In design terms, the two models can yield different duration recommendations once hyperparameters are fit from data, so model selection here directly affects stopping/extension decisions. This distinction has a concrete interpretation in randomized experiments, where a user is typically first exposed to treatment at their first trigger: the first-trigger time underlying the \tgmodel is then determined before any treatment effect can act, whereas the daily activity counts underlying the \bemodel reflect post-trigger behavior and may therefore be altered by treatment. When treatment is expected to change re-triggering dynamics, this gives a principled reason to prefer the first-trigger summary.

We consider next the distribution of the number of new users that were not active in the first $D_0$ days but are active at least once in the following $D_1$ days.
\begin{proposition}\label{prop:pred_unseen}
    Let $N^*_\dayidx$ be the number of users that trigger for the first time on the $d$-th day $(d > D_0)$. 
    Then, conditionally on $\tilde \Delta_{1,h}$ as in \Cref{cor:post_models}, the $N^*_\dayidx$ are independent and $N^*_\dayidx \sim \mathrm{Poi}(\tilde \Delta_{1,h}^{-\alpha}\psi_{r_*}(d-1,1))$, where $r_*=r$ for \nbmodel and $r_*=1$ for the \bemodel and \tgmodel models. For the \bemodel and \tgmodel models, this reduces to $N^*_{D_0+w}\sim\mathrm{Poi}(\alpha\tilde \Delta_{1,h}^{-\alpha}\mathrm{B}(1-\alpha,D_0+w))$ for $w=1,\ldots,D_1$.
    Let $\news{D_0}{D_1} = \sum_{d=D_0+1}^{D_0+D_1} N^*_{d}$ be the number of users that were not active in the first $D_0$ days but are active at least once between days $D_0+1, \ldots, D_0 + D_1$ as per \Cref{eq:news}. 
    Then, conditionally on $\tilde \Delta_{1,h}$, $\news{D_0}{D_1}\sim\mathrm{Poi}(\tilde \Delta_{1,h}^{-\alpha}\psi_{r_*}(D_0,D_1))$, and marginally under all the three models, the posterior distribution of $\news{D_0}{D_1}$ given the data is
    \[
        \news{D_0}{D_1} 
        \sim \mathrm{NegBin}
        \left\{
            N_{D_0} + c  + 1, p_{D_0}^{(D_1)}
        \right\},
    \]
    where 
    \[
        p_{D_0}^{(D_1)} = \frac{\psi_{r_*}(D_0, D_1)}{\beta + \psi_{r_*}(0, D_0 + D_1)}
    \]
    with $r_*=r$ for \nbmodel and $r_*=1$ for the \bemodel and \tgmodel models.
\end{proposition}

\Cref{prop:pred_unseen} is pivotal for our inferential goals, as it allows us to obtain a point estimator for the number of new users that will trigger in the following $D_1$ days by taking the expectation of $\news{D_0}{D_1}$, which reduces to
\begin{equation}\label{eq:u_estim}
\hatnews{D_0}{D_1} = (N_{D_0}+c+1)\frac{p_{D_0}^{(D_1)}}{1-p_{D_0}^{(D_1)}} = (N_{D_0}+c+1)\frac{\psi_{r_*}(D_0,D_1)}{\beta+\psi_{r_*}(0,D_0)}.
\end{equation}

In addition to estimating the number of new users that will be active at least once in a follow-up observation period, another quantity of interest to experimenters is the number of total future re-trigger counts. 
Assuming that the data are as in \Cref{eq:trp1}, where $\daycount$ denotes the number of triggers of the $\useridx$-th user on day $\dayidx$, the total number of triggers in the follow-up period of $D_1$ days is defined as $\newstot{D_0}{D_1} = \sum_{\dayidx=1}^{D_1} \sum_{n \geq 1} A_{D_0 + d, n}$.
This information can be estimated using the \nbmodel model. 

To this end, let $\news{D_0}{D_1, j}$ denote the number of users who have not triggered before day $D_0$ and will trigger exactly $j \geq 1$ times in a follow-up period of $D_1$ days. Namely,
\[
  \news{D_0}{D_1, j} = \sum_{\useridx \geq 1} \mathrm{I}\left[\sum_{\dayidx=1}^{D_0} A_{d,n} = 0 \right]  \mathrm{I}\left[\sum_{\dayidx=1}^{D_1} A_{D_0 + d, n} = j \right].
\]
Similarly, let $S_{D_0}^{(D_1)}$ denote the total re-triggers caused by the previously observed $N_{D_0}$ users in the follow-up period, that is:
\[
    S_{D_0}^{(D_1)} = \sum_{\dayidx=1}^{D_1} \sum_{\useridx = 1}^{N_{D_0}} A_{D_0 + d, n}.
\]
The following proposition characterizes the distribution of re-trigger counts.
\begin{proposition}\label{prop:pred_total}
    Let $Z_{1:D_0} \sim \mbox{\nbmodel}$, let $t_\useridx := \sum_{d=1}^{D_0} A_{d,\useridx}$, $T_0^{(D_0)} := \sum_{\useridx=1}^{N_{D_0}} t_\useridx$, $a := N_{D_0}+c+1$, and $b := \beta+\psi_r(0,D_0)$. Then the posterior of $\newstot{D_0}{D_1}$ given $Z_{1:D_0}$ coincides with the law of $\sum_{j \geq 1} j \news{D_0}{D_1, j} + S_{D_0}^{(D_1)}$, such that, conditionally on $\tilde\Delta_{1,h}$, the $\news{D_0}{D_1,j}$ are independent Poisson random variables with means $\tilde\Delta_{1,h}^{-\alpha}\rho_{D_0}^{(D_1,j)}$. Marginally, they are dependent through $\tilde\Delta_{1,h}^{-\alpha}$ and each has distribution
    \[
        \news{D_0}{D_1, j} \mid Z_{1:D_0} \sim \mathrm{NegBin}(N_{D_0} + c + 1, p_{D_0}^{(D_1, j)}),
    \]
    with 
    \[
        p_{D_0}^{(D_1, j)} = \frac{\rho_{D_0}^{(D_1,j)}}{\beta+\psi_r(0,D_0)+\rho_{D_0}^{(D_1,j)}},  \quad \rho_{D_0}^{(D_1,j)} = \alpha \binom{j+rD_1-1}{j}\mathrm{B}(j-\alpha,r(D_0+D_1)+1).
    \]
    Moreover, $S_{D_0}^{(D_1)} \mid Z_{1:D_0} = \sum_{\useridx = 1}^{N_{D_0}} \mathrm{NegBin}(rD_1, \theta^*_\useridx)$, where $\theta^*_\useridx \sim \mathrm{Beta}(t_\useridx-\alpha,rD_0+1)$ is as in Corollary \ref{cor:post_models}. The Bayesian estimator for $\newstot{D_0}{D_1}$ is then
    \[
        \hatnewstot{D_0}{D_1} = \E\left[\newstot{D_0}{D_1} \mid Z_{1:D_0}\right] = (N_{D_0}+c+1)\frac{\alpha rD_1\mathrm{B}(1-\alpha,rD_0)}{\beta+\psi_r(0,D_0)}+\frac{D_1}{D_0}\{T_0^{(D_0)}-\alpha N_{D_0}\}.
    \]
\end{proposition}
The estimator $\hatnewstot{D_0}{D_1}$ has a natural interpretation: the first term corresponds to the total trigger counts that we expect from the hitherto unseen users. Such a quantity depends exclusively on the hyperparameters of the prior and $N_{D_0}$. The second term instead refers to those users that were already active, and depends on the total number of triggers in the first $D_0$ days. In particular, it consists of a $D_1/D_0$ extrapolation of observed total triggers, whose observed per-day total is $T_{0}^{(D_0)} / D_0$, discounted by $\alpha N_{D_0}$. 
On the other hand, the interpretation of the first term is trickier, as it might be large in two extreme cases, depending on the hyperparameters of the prior: a large number of new users appear, but each of them triggers seldom, or a small number of new users become active, but each of them triggers a lot. 
Since we will adopt an empirical Bayesian approach to select adequate values of the hyperparameters, our estimator is entirely data-driven and we don't have to specify in advance what we believe the future number of users and their propensity to click will be.

\subsection{Estimating the number of days to reach a participation threshold}\label{sec:est_days}

We consider now a related but slightly different prediction problem: suppose that we would like for the A/B test to reach a target level of traffic, so that $M$ users are exposed to the intervention before it is terminated. We will now show how our models can be used to estimate the number of days needed to reach this inclusion level, based on pilot data.

\begin{figure}[H]
    \centering
    \includegraphics[width=0.5\textwidth]{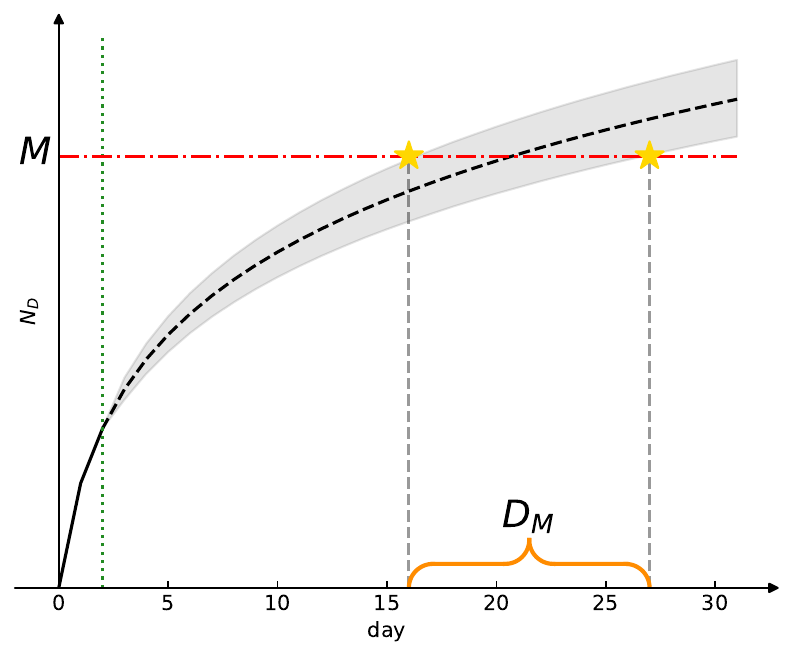}
    \caption{Inversion technique to estimate $D_M$. The solid line represents the data. Dashed line and shaded area are the mean and the global credible band of $N_{D_0} + \news{D_0}{\ell}$, where $\ell$ varies across the horizontal axis. 
    The interval for $D_M$ (curly bracket) is obtained by slicing the grey area at $M$.}
    \label{fig:inversion_ci}
\end{figure}

For a target $M>N_{D_0}$, let
\[
    D_M := \inf\{\ell\ge 1:N_{D_0}+\news{D_0}{\ell}\ge M\}
\]
denote the number of follow-up days (beyond the pilot) needed to reach $M$ total users. Equivalently, $D_M$ is the time needed to observe $M-N_{D_0}$ post-pilot users, i.e., users that did not trigger during the pilot period $1, \ldots, D_0$.
One can approximate the posterior law of $D_M$ directly by Monte-Carlo; see Supplementary \ref{app:est_days_direct}. Here, we focus on a more scalable heuristic obtained by ``slicing'' a global prediction band for the trajectory of the number of users that trigger before day $\ell$, i.e., $N_{D_0} + \news{D_0}{\ell} =: N_{D_0+\ell}$, as depicted in \Cref{fig:inversion_ci}.

To construct a global credible band for $(N_{D_0 + \ell})_{\ell \geq 1}$, we work with the random variables $N^*_\dayidx$ introduced in \Cref{prop:pred_unseen}. Indeed, for any $\ell$, $(N_{D_0 + 1}, \ldots, N_{D_0 + \ell})$ and $(N^*_{D_0+1}, \ldots, N^*_{D_0+\ell})$ are in one-to-one correspondence given the observed data.
In particular, we aim to construct a global credible band of the form $\{(n^{lo}_{D_0+\ell}, n^{hi}_{D_0 + \ell})\}_{\ell = 1}^{D^{up}}$ with simultaneous posterior coverage at least $1-\varepsilon$, so that $\text{Pr}(n^{lo}_{D_0 + \ell} \leq N_{D_0 + \ell} \leq n^{hi}_{D_0 + \ell} \text{ for all } \ell \mid Z_{1:D_0}) \geq 1 - \varepsilon.$
First we simulate $Q$ times from the posterior law of $\tilde \Delta_{1,h}^{-\alpha}$ as in \Cref{cor:post_models} and, conditional on $\tilde \Delta_{1,h}^{-\alpha}$, we sample the values $N^*_{D_0+\ell}$, $\ell = 1, \ldots, D^{up}$ as in \Cref{prop:pred_unseen}, keeping only the $(1-\varepsilon)Q$ tuples of $(\tilde \Delta_{1,h}^{-\alpha}, \{N^*_{D_0+\ell}\}_{\ell=1}^{D^{up}})$ with highest posterior density.
Then, we evaluate the $(1-\varepsilon)Q$ trajectories for 
$N_{D_0 + \ell} = N_{D_0} + \sum_{j \leq \ell} N^*_{D_0+j}$, denoted by $\{(N_{D_0 + 1}^{(k)}, \ldots, N_{D_0 + D^{up}}^{(k)}), k = 1, \ldots, (1-\varepsilon)Q\}$ and set
\[
   n^{lo}_{D_0 + \ell} = \min_{k} N^{(k)}_{D_0 + \ell}, \qquad  n^{hi}_{D_0 + \ell} = \max_{k} N^{(k)}_{D_0 + \ell}.
\]
We found that the interval for $D_M$ is robust to the chosen value of $D^{up}$ if this is sufficiently large. We always fix $D^{up} = 3 \mathring D_M$ in our simulations, where $\mathring D_M$ is a point estimate for $D_M$ defined by $N_{D_0}+\hatnews{D_0}{\mathring D_M} \approx M$.

\section{Numerical Implementation} \label{sec:numerical_implementation}

The theoretical results derived in Section \ref{sec:bnp_method} yield closed-form expressions for Bayesian point estimators and posterior predictive summaries.
Crucially, the expressions found depend on key parameters $(\alpha, c, \beta)$ in the case of \bemodel and \tgmodel models and $(\alpha, c, \beta, r)$ for the \nbmodel model.
Posterior credible intervals can also be computed from the predictive distributions, although the empirical analysis below focuses on point forecasts and duration-planning accuracy.

\begin{figure}[t]
    \centering
    \includegraphics[width=\linewidth]{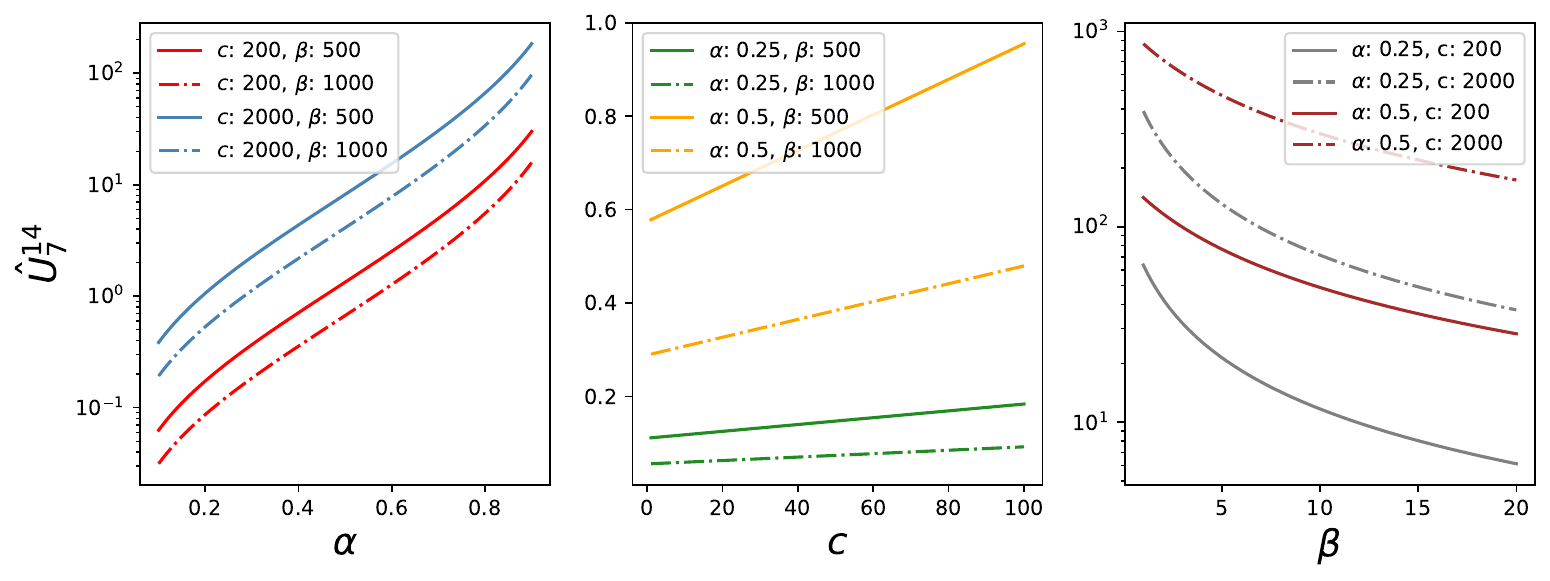}
    \caption{Values of $\hatnews{7}{14}$ for different choices of parameters $(\alpha, c, \beta)$ and $r=1$}
    \label{fig:u_vals}
\end{figure}

Inspecting, for instance, the expectation of $\news{D_0}{D_1}$ in \Cref{eq:u_estim} sheds light on the role played by the different parameters, cf. \Cref{fig:u_vals}. However, it is clear that there is a strong interplay between such parameters, making prior elicitation hard. For instance, $\news{D_0}{D_1}$ is an increasing function of $\alpha$ and $c$ and a decreasing function of $\beta$.
Moreover, prior elicitation must be carried out experiment-by-experiment, which is time-consuming and infeasible in large-scale applications of A/B tests, where thousands of tests are performed daily.
Hence, we take an empirical Bayesian approach and estimate such parameters from the available data.
Compared to being ``fully Bayesian'', i.e., assuming prior distributions for such hyperparameters, our strategy is more scalable as it does not require performing Markov chain Monte-Carlo simulations.

\subsection{Parameter estimation by maximum marginal likelihood}

A natural option to estimate unknown parameters involved in the likelihood and prior distributions is maximizing the observed data's marginal likelihood. See, e.g., \cite{carlin2000empirical}. To this end, the next result gives an explicit expression for the marginal distribution of the data.

\begin{theorem}[Unified marginal likelihood for the SSP models]\label{thm:marg_general}
Let $N=N_{D_0}$ be the number of active users, with labels $\omega^*_1, \ldots, \omega^*_{N}$. Define the observed pilot data by
\[
    \mathcal D_{D_0}
    =
    \begin{cases}
        Z_{1:D_0}, & \text{under \nbmodel or \bemodel,}\\
        Z^{\FT}_{D_0}:=\displaystyle\sum_{n=1}^{N}F_n\delta_{\omega_n^*},
        & \text{under \tgmodel,}
    \end{cases}
\]
where, in the \tgmodel case, $F_n\in\{1,\ldots,D_0\}$ is the first-trigger time of active user $n$. The \tgmodel likelihood can equivalently be evaluated from the daily first-trigger counts $C_d^{\FT}=\sum_{n=1}^{N}\mathbf 1\{F_n=d\}$, $d=1,\ldots,D_0$. After integrating out the user labels, the marginal distribution of $\mathcal D_{D_0}$ is
\begin{equation}\label{eq:marg}
    \prob(\mathcal D_{D_0}) = \frac{\alpha^N \beta^{c+1} \Gamma(N+c+1)}{\Gamma(c+1)\{\beta+\psi_{r_*}(0,D_0)\}^{N+c+1}}\prod_{n=1}^N \Theta_n,
\end{equation}
where $r_*=r$ under \nbmodel and $r_*=1$ under \bemodel and \tgmodel. The $\Theta_n$ terms are
\[
    \Theta_n = \begin{cases}
        \left[
        \prod_{d=1}^{D_0} \binom{A_{d, n} + r - 1}{A_{d, n}}
        \right]
        \mathrm{B}(t_n-\alpha,rD_0+1),\quad
        t_n=\sum_{d=1}^{D_0}A_{d,n}
        & \text{under \nbmodel,} \\
        \mathrm{B}(b_n - \alpha, D_0 - b_n + 1),\quad
        b_n=\sum_{d=1}^{D_0}A_{d,n}
        & \text{under \bemodel,} \\
        \mathrm{B}(1 - \alpha, F_n) & \text{under \tgmodel.}
    \end{cases}
\]
\end{theorem}

Defining the marginal likelihood for parameters $\alpha, c, \beta$, and possibly $r$, as $\mathcal L(\alpha, c, \beta, r;\mathcal D_{D_0}) = \prob(\mathcal D_{D_0}; \alpha, c, \beta, r)$, the empirical Bayesian strategy is to set such parameters as
\begin{equation}
    \hat \alpha, \hat c, \hat \beta, \hat r = \argmax_{\alpha, c, \beta, r} \log \mathcal L(\alpha, c, \beta, r).
    \label{eq:ml}
\end{equation}
The optimization problem is nonconvex in its variables. Therefore, we adopt the \texttt{scipy} implementation of the differential evolution algorithm \citep{storn1997differential}, which is a derivative-free global optimization algorithm.

\subsection{Parameter estimation by curve fitting}

Another option is to frame parameter elicitation as a curve fitting problem for the temporal trajectory of a statistic of interest.
For example, we might want to predict the number of future users first triggering in the next $D_1$ days, using the point estimator of \Cref{eq:u_estim}. We fit the hyperparameters by solving the following regression problem:
\begin{equation}
    \hat \alpha, \hat c, \hat \beta, \hat r = \argmin_{\alpha, c, \beta, r} \sum_{d = 1}^{D_0 - d_0} \left\{ \hatnews{d_0}{d} - u_{d_0}^{(d)} \right\}^2,
    \label{eq:regression}
\end{equation}
where $u_{d_0}^{(d)}$ denotes the true (observed) number of new distinct users observed between day $d_0$ and day $d_0+d$, given $Z_{1:D_0}$. 
Here, $1\le d_0 < D_0$, and $d_0 + d \le D_0$.
In our experiments, we find that the choice of $d_0 = 1$ works well, although in different applications, different choices might work better in practice (e.g., \citet[Section 4]{masoero2022more} recommends $d_0=\lfloor 2/3 \times D_0 \rfloor$).

The curve fitting approach offers two advantages compared to the maximum marginal likelihood one. First, it is much faster to evaluate numerically $\hatnews{d_0}{d}$ compared to the log-likelihood, especially when the number of active users in the first $D_0$ days is large. 
% Hence, curve fitting is typically faster than maximum likelihood.
Second, the objective in \eqref{eq:regression} depends on the data only through the cumulative first-trigger trajectory $(N_d)_{d=d_0}^{D_0}$ (equivalently, the daily first-trigger counts), rather than through the full user-level activity matrix.
Hence, curve fitting is feasible even when only first trigger counts, and not re-trigger counts, are available. 
This scenario might occur on proprietary or large datasets, where, e.g., for privacy reasons or reasons of scale, only aggregate summary statistics can be stored (e.g., the ASOS dataset later analyzed \citep{liu2021datasets}).
% However, in our experience, maximum likelihood typically performs better in terms of predictive accuracy.
Throughout the empirical sections, we measure unseen-user prediction accuracy using the metric of \citet{camerlenghi2022scaled},
\begin{equation}\label{eq:acc_metric}
    v_{D_0}^{(D_1)} := 1 - \min\left\{ \frac{|\news{D_0}{D_1}- \hatnews{D_0}{D_1}|}{\news{D_0}{D_1}}, 1\right\},
\end{equation}
so that $v_{D_0}^{(D_1)}$ equals one under perfect prediction and degrades to zero as the relative prediction error grows.
Supplementary \Cref{app:param_empirics} reports numerical evidence comparing maximum marginal likelihood and curve fitting for parameter estimation.

\section{Simulation studies} \label{sec:simulations}

We now evaluate the performance of our Bayesian nonparametric estimators via simulations on synthetic data, shedding light on several aspects of our models.
In Appendix~\ref{app:simu_true_model}, we verify that our predictors recover the true hyperparameters and produce accurate forecasts when data are drawn from the assumed model.
% Here we focus on a more challenging setting in which data are drawn from a Zipfian model (details in \Cref{sec:zipf_simu}) and on the $D_M$ interval comparison.

\subsection{Prediction accuracy when the data are drawn from a Zipfian model}\label{sec:zipf_simu}

We consider synthetic data from a Zipf-Poisson distribution. 
We evaluate the performance of our predictors against that of a collection of competing methods: the Jackknife estimator of order $k$ (J$k$) \citep{gravel2014predicting}, the linear programming approach of \cite{zou2016quantifying}, the Good--Toulmin estimator \citep{good1956number}, the Bayesian models of \citet{ionita2009estimating} (BB) and \citet{richardson22a} (HBG), and the three-parameter Indian buffet process \citep[IBP,][]{Teh(09)}.
In particular, for a given cardinality $N_{\infty} = 1{,}000{,}000$ of possible users we draw trigger data as follows. We endow each user $\useridx$ with a triggering rate $\theta_\useridx = \useridx^{-\tau}$, with $\tau = 0.6, 0.7, 0.8, 0.9$. Then, for every experimental day $\dayidx$, we determine whether each user $\useridx$ re-triggers by flipping a Bernoulli coin, $X_{\dayidx, \useridx} \sim \mathrm{Bernoulli}(\theta_\useridx)$. Conditionally on $X_{\dayidx, \useridx} = 1$, we sample $Z_{\dayidx, \useridx} \sim \mathrm{tPoisson}(1+m_{\dayidx-1, \useridx}/\dayidx; 0)$, where $m_{\dayidx-1, \useridx}$ is the total re-trigger count of unit $\useridx$ up to day $\dayidx-1$ and $\mathrm{tPoisson}(a;b)$ is the law of a Poisson random variable with parameter $a$ and supported on $\{b+1, b+2, \ldots \}$. For each value $\tau$, we generate $M=100$ datasets. We retain only $\pilotdays=5$ days for training, and compute the predictive accuracy $v_{\pilotdays}^{(\followupdays)}$ for $\followupdays = 50$ days ahead. Our results in \Cref{fig:synthetic_accuracy_zipf} show that across values of $\tau$ the \nbmodel, \tgmodel, \bemodel, and IBP provide the best performance. 
Each boxplot reports the median of the accuracy, with boxes spanning the interquartile range.

\begin{figure}[t]
    \centering
    \includegraphics[width=\linewidth]{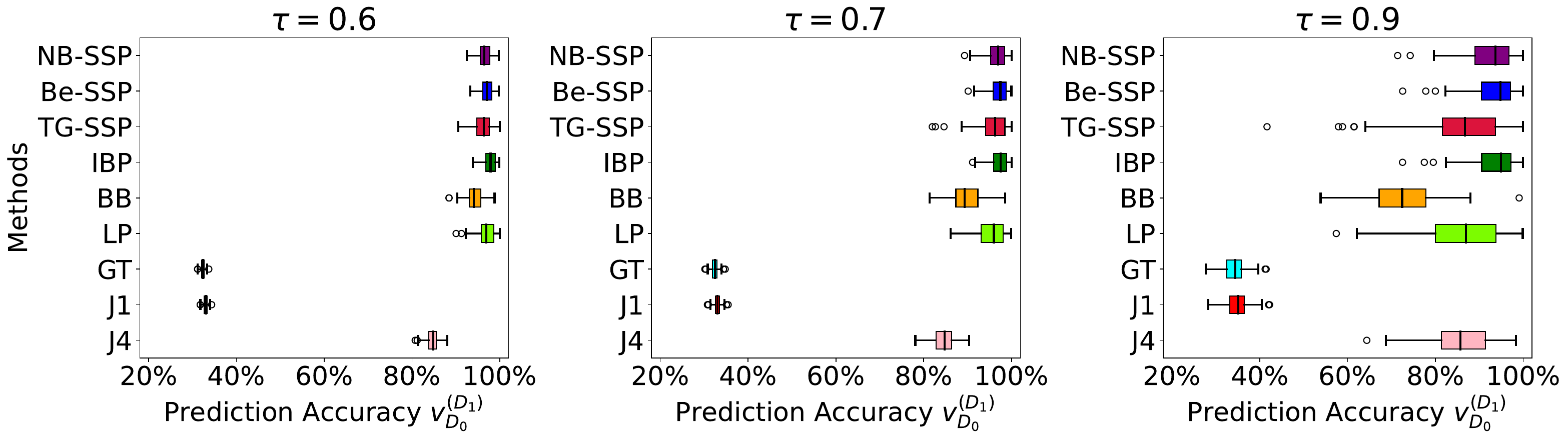}
    \caption{Prediction accuracy $v_{\pilotdays}^{(\followupdays)}$ of predictors $\hatnews{\pilotdays}{\followupdays}$ on synthetic data from the Zipfian model for different choices of the parameter $\tau$.}
    \label{fig:synthetic_accuracy_zipf}
\end{figure}

Moreover, in addition to providing estimates for $U_{D_0}^{(D_1)}$, the \nbmodel can be leveraged to form predictions for the total number of triggers ${T}_{\pilotdays}^{(\followupdays)}$. We show in \Cref{fig:synthetic_accuracy_sums} in Appendix \ref{app:plots} that the total sum predictor $\hat{T}_{\pilotdays}^{(\followupdays)}$ achieves great accuracy in predicting the total sum via a survival plot. 
While beyond the scope of the present paper, total re-trigger counts can be leveraged in forming better estimates of long-term treatment effects and to better predict the duration needed for an experiment to deliver significant results \citep{richardson22a,wan2023experimentation}.
To this end, we define a notion of accuracy for the prediction of the total re-trigger counts $\hat{T}_{\pilotdays}^{(\followupdays)}$:
\[
% \begin{equation*}
    \tilde{v}_{\pilotdays}^{(\followupdays)} ~:=~ 1 - \min\left\{ \frac{|t_{\pilotdays}^{(\followupdays)}-\hat{T}_{\pilotdays}^{(\followupdays)}|}{t_{\pilotdays}^{(\followupdays)}}, 1\right\} \in [0,1].
\]
$t_{\pilotdays}^{(\followupdays)}$ is the observed total re-trigger count between $\pilotdays$ and $\pilotdays+\followupdays$

\subsection{Estimating the days to a given participation threshold}\label{sec:interval_comparison}

We consider here the estimation of $D_M$ and compare the inversion approach described in \Cref{sec:est_days} with the direct posterior sampler detailed in Supplementary \Cref{app:est_days_direct}.
We consider only the \tgmodel model.
Analogous results hold for the other models presented in the paper.

We simulate data $X_{d, n}$ as in \Cref{sec:zipf_simu}, for $\tau = 0.8, 1.0, 1.2, 1.4$.
We generate data for $D_0=14$ days, thus observing $N_{D_0}$ users, and we want to estimate the number of days needed to reach a total target of $M=\eta N_{D_0}$ users for $\eta = 1.5, 2.0, 5.0$.

We compare the intervals for $D_M$ based on the inversion technique depicted in \Cref{fig:inversion_ci} with the posterior mean and 95\% credible intervals obtained from the posterior of $D_M$,
which is approximated using $1{,}000$ draws from \Cref{algo:simulate_D_M}.
The intervals obtained via the inversion technique are wider than those based directly on the posterior of $D_M$; see \Cref{fig:interval_lengths}.
Their length increases with the tail parameter $\tau$ of the data-generating process.
This is intuitive: smaller values of $\tau$ lead to faster early user arrivals, while larger values can require hundreds or thousands of days to reach the target participation level.

\begin{figure}[t]
    \centering
    \includegraphics[width=\linewidth]{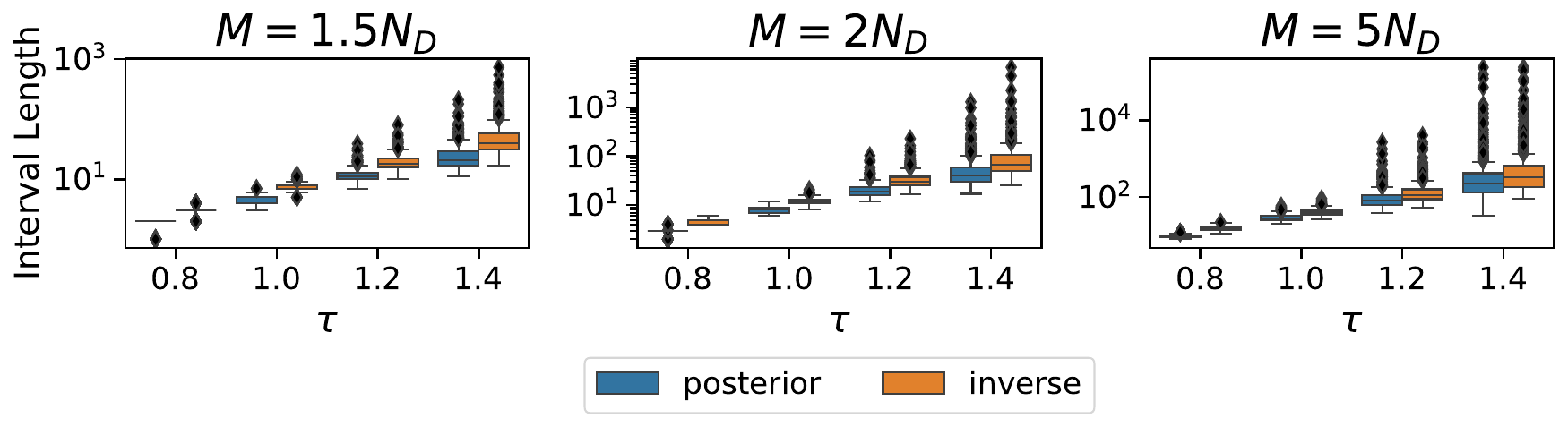}
    \caption{Length of posterior and inversion-based intervals for $D_M$ for 500 simulated datasets. Different panels correspond to different target values $M$; within each panel, the tail parameter $\tau$ varies along the horizontal axis.}
    \label{fig:interval_lengths}
\end{figure}

For fixed numbers of Monte Carlo draws $K$ and $Q$, both procedures have worst-case computational cost linear in $D^{up}$. Algorithm~\ref{algo:simulate_D_M} simulates at most $D^{up}$ daily arrival counts per posterior draw, whereas the inversion approach simulates and evaluates $Q$ trajectories of length $D^{up}$. Their realized runtimes can nevertheless differ because of implementation constants and the number of follow-up days visited before the target is reached. To illustrate this, we performed a small timing experiment using data generated from (DG2) with $D_0=7$, $\beta=0.5$, $c=1000$, and $\alpha\in\{0.25,0.5,0.75\}$, and computed inversion-based and posterior intervals for $D_M$ with $M=2N_{D_0}$. Computing the inversion intervals took around $0.1$s across all simulations, whereas the posterior intervals took from $0.2$s at $\alpha=0.25$ to almost $10$s at $\alpha=0.75$.

\section{Real data analyses} \label{sec:exp_real}

% === NEW Section 6 (v2) ===
% Addresses: AE ("data analysis thin"), R2-b (practical impact),
% R2-j (dataset description), T1.b (case studies), T2.e (decisions)

We evaluate all methods on the four datasets introduced in \Cref{sec:setting}. $\pilotdays = 7$ days of pilot data are used for fitting; predictions target the follow-up period.

Table~\ref{tab:cross_accuracy} reports median prediction accuracy $v$, defined in \Cref{eq:acc_metric}, across datasets.
\nbmodel performs best on UCI and the REES46 28-day non-overlapping benchmark, and is competitive on the REES46 $k=21$ benchmark, where full trigger-count data is available.
The $k=21$ and $k=100$ benchmarks use rolling windows with follow-up lengths of 21 and 100 days respectively, constructed with stride~1 (one experiment per start day; see Appendix~\ref{app:benchmark_construction} for the full specification).
We report $k=21$ for the accuracy comparison (Table~\ref{tab:cross_accuracy}) and $k=100$ for the hitting-time analysis (Table~\ref{tab:ht_main}); the Supplement additionally reports $k=50$.
On ASOS, which provides only first-trigger counts, \tgmodel---fitted via maximum marginal likelihood of the geometric model---outperforms the other methods.
Performance varies across datasets. \nbmodel leads two of the four benchmarks and has the highest unweighted mean of the benchmark-level median accuracies. IBP leads on REES46 $k=21$, and \tgmodel leads on ASOS.%

\begin{table}[t]
\centering
\caption{Median prediction accuracy $v$, defined in \Cref{eq:acc_metric}, across datasets. Bold = best per row; ties are bolded.}
\label{tab:cross_accuracy}
\small
\begin{tabular}{@{}lrccccc@{}}
\toprule
& $n$ & $\followupdays/\pilotdays$
& \tgmodel & IBP & \nbmodel (reg) \\
\midrule
% Values regenerated from the committed reproduction package
% (paper-code: paper_fitting/values/cross_dataset_accuracy.csv, seed 0).
% ASOS TG-SSP is the geometric marginal-likelihood fit. REES46 (28-day) TG-SSP is
% bimodal across seeds; the typical (low) mode is reported -- see the appendix.
UCI & 13 & 3.0 & 0.70 & 0.80 & \textbf{0.84} \\
REES46 (28-day) & 7 & 3.0 & 0.47 & 0.81 & \textbf{0.90} \\
REES46 ($k{=}21$) & 186 & 3.0 & 0.80 & \textbf{0.81} & 0.80 \\
ASOS & 144 & 7.9 & \textbf{0.72} & 0.16 & 0.50 \\
\bottomrule
\end{tabular}
\end{table}

ASOS poses the longest extrapolation horizons ($\followupdays$ up to 239 days; median $\followupdays/\pilotdays=7.9$) and provides only first-trigger data.
IBP and \nbmodel have lower median accuracy on this benchmark.
\tgmodel achieves the highest median accuracy (0.72) on ASOS, consistent with its design for first-trigger-time data; \nbmodel (0.50) is fitted via curve regression, which partly explains its lower accuracy on this dataset.

Figure~\ref{fig:case_uci} shows prediction trajectories on three UCI experiments spanning moderate ($N_{\pilotdays} = 249$), medium (348), and large (474) pilot sizes.
\nbmodel tracks the truth most closely: on Exp~2, it predicts 508 new users by day 28 (truth: 502, $v = 0.99$); on Exp~6, 669 new users (truth: 592, $v = 0.87$).
This is directly actionable: it tells the experimenter whether the experiment is on track to reach its participation target.

\begin{figure}[t]
  \centering
  \includegraphics[width=\linewidth]{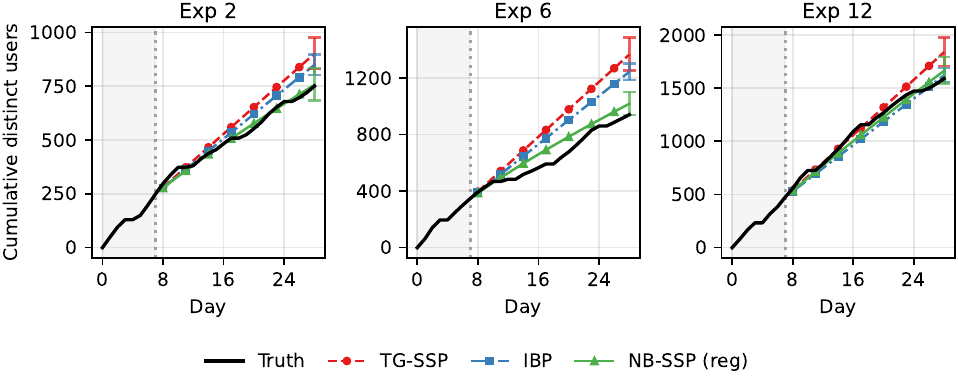}
  \caption{UCI: predicted cumulative new users vs.\ truth for three experiments.
  Shaded region = pilot. \nbmodel tracks truth most closely.}
  \label{fig:case_uci}
\end{figure}

\subsection{Duration-to-target decisions}

We now translate the forecasting task into a concrete experiment-planning decision.
Suppose that, after observing a pilot of length $\pilotdays$, the experimenter wants to choose how many additional days are needed to reach a target number of distinct users.
For a target multiplier $\eta>1$, define the target participation level as
\[
    M_\eta = \lceil \eta N_{\pilotdays} \rceil.
\]
Define $D_\eta := \pilotdays + D_{M_\eta}$ as the absolute experiment day on which the cumulative number of distinct users first reaches $M_\eta$. Thus $D_{M_\eta}$ is the number of additional post-pilot days, whereas $D_\eta$ is the total experiment duration. Operationally, $D_\eta$ is the duration recommendation produced from the pilot: if the estimated value is small, the experiment can be planned as a short run; if it is large, the experimenter knows early that the test should be extended or redesigned.

For each fitted method, we estimate $D_{M_\eta}$ as discussed in Section~\ref{sec:est_days} and set $\widehat D_\eta=\pilotdays+\widehat D_{M_\eta}$.
As a simple benchmark, we also consider a linear extrapolation rule that assumes the average distinct-user arrival rate observed during the pilot remains constant.
Writing $\widehat r_{\mathrm{lin}} = N_{\pilotdays}/\pilotdays$, this benchmark sets
\[
    \widehat D_{\eta}^{\mathrm{lin}}
    =
    \pilotdays
    +
    \left\lceil
    \frac{M_\eta - N_{\pilotdays}}{\widehat r_{\mathrm{lin}}}
    \right\rceil.
\]
We compare this recommendation with the realized hitting time computed from the full experiment window.
This gives a decision-level error, measured in days, rather than only a prediction error for a fixed horizon.

\begin{table}[t]
\centering
\caption{Hitting-time MAE (days): mean absolute error $|\hat D_\eta - D_\eta|$ across experiments.
Each column uses the common complete-case set for which the target is observed and all five methods produce a hitting-time prediction; sample sizes are shown below the target levels.
Bold = best per column; ties are bolded.
\nbmodel achieves the lowest MAE on ASOS at every $\eta$ level, reducing ASOS planning error by about \htReductionAsosMin--\htReductionAsosMax\% relative to the next-best method.}
\label{tab:ht_main}
\small
\begin{tabular}{@{}l ccc ccc@{}}
\toprule
& \multicolumn{3}{c}{\textbf{ASOS}}
& \multicolumn{3}{c}{\textbf{REES46} $k{=}100$} \\
\cmidrule(lr){2-4} \cmidrule(lr){5-7}
Method & $\eta{=}1.5$ & $\eta{=}2$ & $\eta{=}3$
       & $\eta{=}1.5$ & $\eta{=}2$ & $\eta{=}3$ \\
& $(n{=}\htNAsosA)$ & $(n{=}\htNAsosB)$ & $(n{=}\htNAsosC)$
& $(n{=}\htNReesA)$ & $(n{=}\htNReesB)$ & $(n{=}\htNReesC)$ \\
\midrule
\nbmodel (reg)
  & \textbf{\htMAEAsosNBSSPA} & \textbf{\htMAEAsosNBSSPB} & \textbf{\htMAEAsosNBSSPC}
  & \textbf{\htMAEReesNBSSPA} & \textbf{\htMAEReesNBSSPB} & \textbf{\htMAEReesNBSSPC} \\
\tgmodel
  & \htMAEAsosTGSSPA & \htMAEAsosTGSSPB & \htMAEAsosTGSSPC
  & \textbf{\htMAEReesTGSSPA} & \textbf{\htMAEReesTGSSPB} & \htMAEReesTGSSPC \\
IBP
  & \htMAEAsosIBPA & \htMAEAsosIBPB & \htMAEAsosIBPC
  & \htMAEReesIBPA & \htMAEReesIBPB & \htMAEReesIBPC \\
Linear extrap.
  & \htMAEAsosLinA & \htMAEAsosLinB & \htMAEAsosLinC
  & \textbf{\htMAEReesLinA} & \htMAEReesLinB & \htMAEReesLinC \\
Jackknife (J3)
  & \htMAEAsosJackA & \htMAEAsosJackB & \htMAEAsosJackC
  & \textbf{\htMAEReesJackA} & \textbf{\htMAEReesJackB} & \htMAEReesJackC \\
\bottomrule
\end{tabular}
\end{table}

Table~\ref{tab:ht_main} reports the mean absolute error (MAE) in predicted hitting time for $\eta \in \{1.5, 2.0, 3.0\}$.
For each dataset and target level, we evaluate every method on the same complete-case set: the observed trajectory reaches the target and all five methods return a finite hitting-time prediction.
On ASOS, \nbmodel predicts the time to reach $2\times$ pilot participation within \htMAEAsosNBSSPB~days on average, compared to \htMAEAsosTGSSPB--\htMAEAsosLinB~days for competitors.
At $\eta = 3$, the advantage is even larger: \htMAEAsosNBSSPC~days vs.\ \htMAEAsosIBPC--\htMAEAsosLinC~days, a \htReductionAsosMax\% reduction in planning error relative to IBP, the next-best method.
On REES46, all methods perform well (MAE $\leq 2.8$~days at $\eta = 3$), reflecting the large pilot samples ($\sim$1M users); \nbmodel has the lowest MAE or is tied for the lowest MAE at every $\eta$ level.
Figure~\ref{fig:ht_scatter_main} illustrates the decision impact on a representative ASOS experiment (arm ee6ff7\_C, $\followupdays = 46$, $N_{\pilotdays} = 448{,}882$): linear extrapolation recommends planning for day~14, while \nbmodel recommends day~28 and the true hitting time is day~31.
An experimenter following the linear recommendation would have stopped 17~days too early, falling short of the target sample size.
The accompanying boxplots show the common-set planning-error distributions for the ASOS arms (center) and the REES46 experiments (right); in both datasets, \nbmodel has the tightest distribution centered near zero.
An ASOS arm is \emph{censored} when the $\eta$-target is not reached within its observed follow-up, so its realized $D_\eta$ is right-censored.
We exclude such arms and any arm for which at least one method does not return a finite prediction, yielding common-set sample sizes of \htNAsosA, \htNAsosB, and \htNAsosC\ for $\eta=1.5$, $2$, and $3$, respectively.

\begin{figure}[t]
  \centering
  \includegraphics[width=\linewidth]{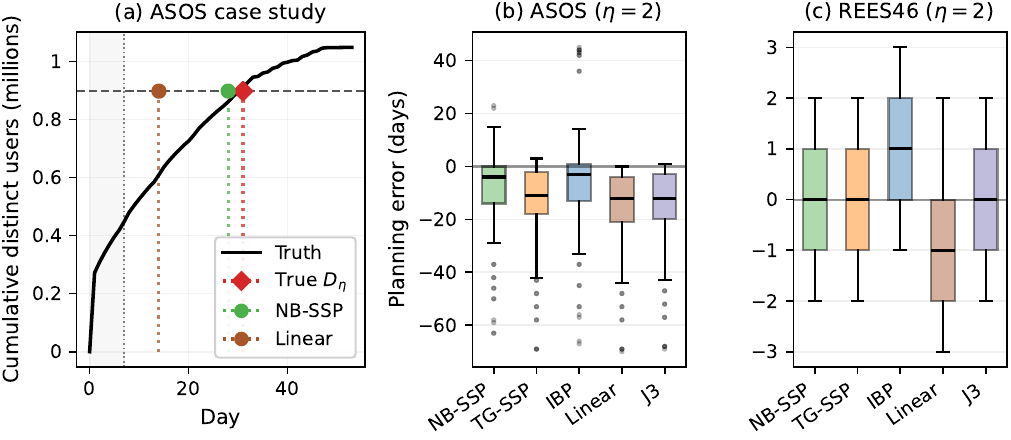}
  \caption{Duration-to-target decision analysis.
  \emph{Left:} case study (arm ee6ff7\_C, $\followupdays = 46$): cumulative user curve with $\eta = 2$ target.
  \emph{Center:} planning error $\hat D_\eta - D_\eta$ across the \htNAsosB\ common complete-case ASOS arms ($\eta = 2$).
  \emph{Right:} the corresponding planning error across REES46 experiments.}
  \label{fig:ht_scatter_main}
\end{figure}

\paragraph*{Remark (Connection to statistical power)}
Consider a two-arm experiment with equal allocation ($M_\eta / 2$ users per arm), minimum detectable effect magnitude $\delta>0$, within-arm standard deviation $\sigma$, and significance level $\alpha_{\mathrm{test}}$.
Under a standard two-sample $z$-test, the power is approximately $\Phi\!\left(\frac{\sqrt{M_\eta}}{2}\frac{\delta}{\sigma}-z_{1-\alpha_{\mathrm{test}}/2}\right)$; see Appendix~\ref{app:power_derivation} for a derivation and \citet{Gualavisi2025} for a detailed treatment in the A/B testing context.
Thus the hitting-time MAE directly quantifies the accuracy of power-aware duration planning: an experimenter who specifies $\delta$, $\sigma$, $\alpha_{\mathrm{test}}$, and a desired power level can compute the required $M_\eta$, then use the hitting-time prediction to plan the experiment duration.

\subsection{Proprietary data}

The proprietary collection contains 759 production A/B tests comprising 1,774 arms. For the new-user accuracy benchmark, we evaluate one eligible control arm per test.
The proprietary dataset contains daily trigger counts; parameter estimation via maximum marginal likelihood is feasible for \bemodel and \tgmodel and via curve fitting for \nbmodel.
We also fit the beta-binomial \citep{ionita2009estimating} (BB) and Indian buffet process \citep{masoero2022more} (IBP) models via maximum marginal likelihood.
The right panel of \Cref{fig:accuracy_1} reports the ASOS first-trigger benchmark from the public analysis above for comparison; for ASOS, we fit the Bayesian models via maximum marginal likelihood or curve fitting.
For the proprietary dataset, $\pilotdays = 7$ and we predict at $\followupdays = 21$; for ASOS, $\pilotdays = 7$ and we predict at the last available time point.

\begin{figure}
    \centering
    \begin{subfigure}[h]{0.66\linewidth}
        \includegraphics[width=\linewidth]{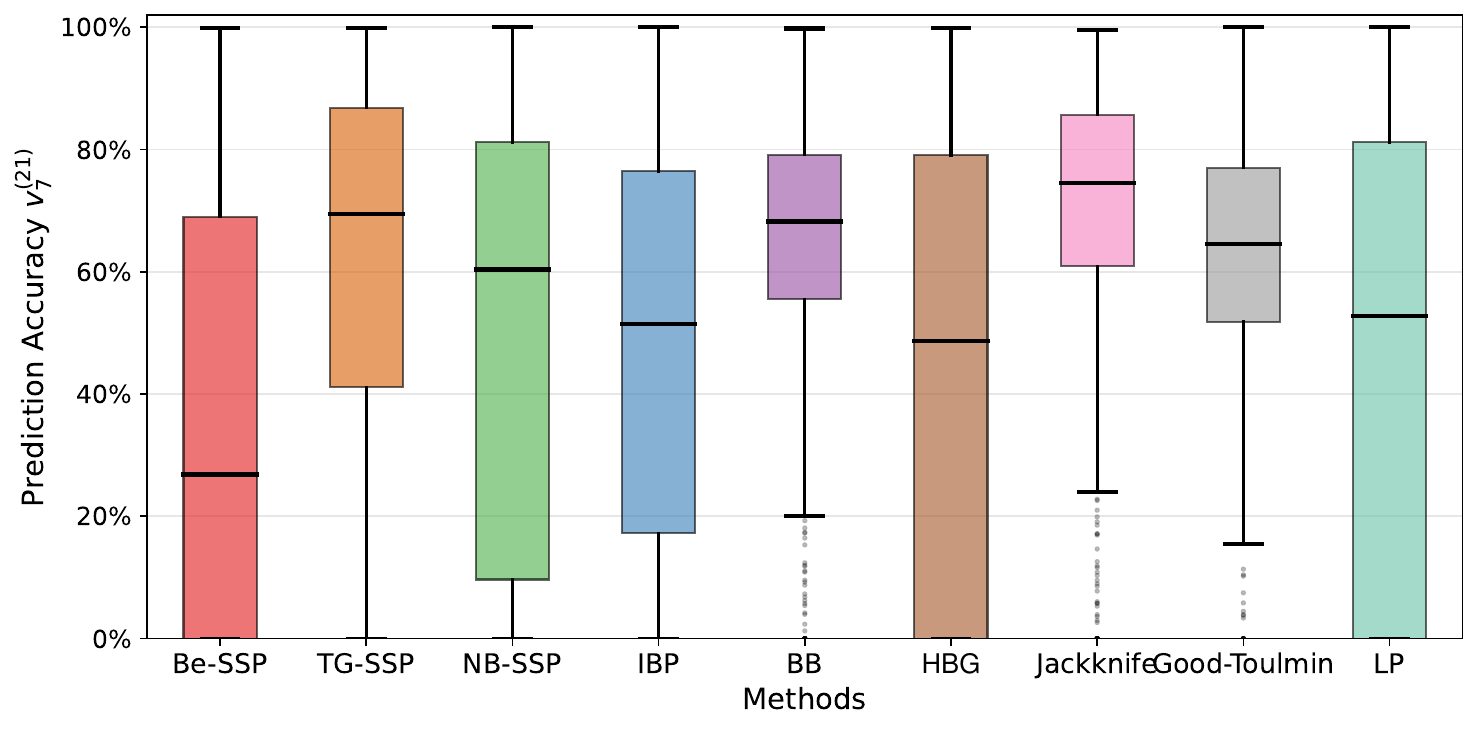}
    \end{subfigure}%
    \begin{subfigure}[h]{0.33\linewidth}
        \includegraphics[width=\linewidth]{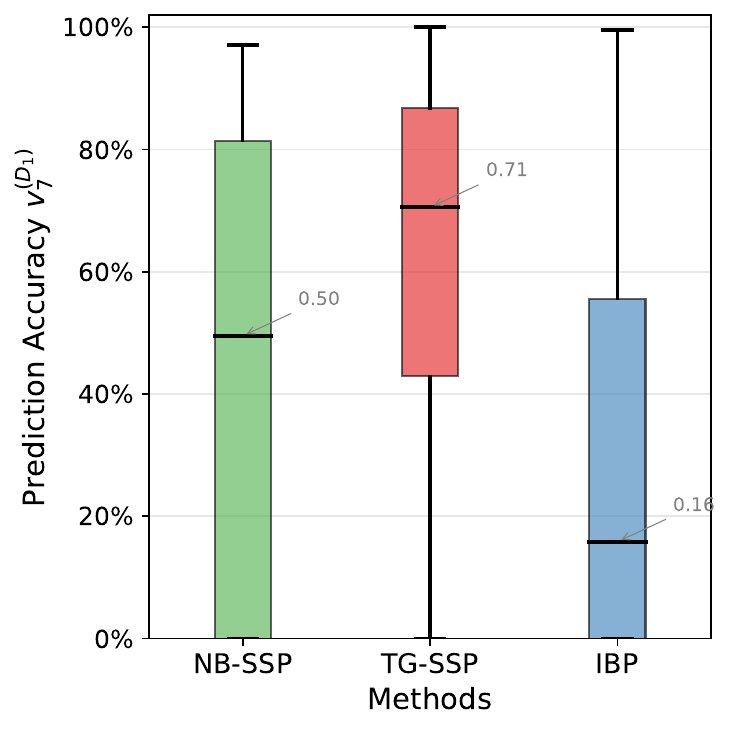}
    \end{subfigure}
    \caption{Prediction accuracy on proprietary (left) and ASOS (right) data.}
    \label{fig:accuracy_1}
\end{figure}

Last, we assess in \Cref{fig:true_accuracy_sums} in Appendix \Cref{app:plots} the performance on a subset of 50 experiments at predicting total future triggers, retaining $\pilotdays=7$ days and extrapolating at $\followupdays=14, 21, 28, 35$.

\section{Discussion} \label{sec:discussion}
 
This paper presents a novel Bayesian nonparametric framework for predicting user engagement in online A/B tests using stable beta-scaled process priors. While the theoretical foundation draws from rather complex random measure theory, the resulting methodology yields tractable, closed-form estimators that can be readily implemented by practitioners. This accessibility is particularly valuable in industrial settings, where rapid decision-making about experiment duration and resource allocation is critical.

A key strength of our approach lies in its computational efficiency. Indeed, the posterior distributions and predictive quantities have explicit expressions that avoid the need for costly MCMC sampling. The hyperparameter optimization, performed through maximum marginal likelihood or curve fitting, was computationally practical in our experiments. The only material instability we observed occurred for \tgmodel on the seven-window REES46 benchmark, where the marginal-likelihood solutions were bimodal across random restarts; see \Cref{app:tg_optimization}. We therefore use multiple random restarts and report the sensitivity analysis for this case. This combination of theoretical rigor and computational tractability makes our method particularly appealing for large-scale applications where thousands of experiments may need to be analyzed simultaneously.

Our empirical results demonstrate that the proposed models -- \bemodel, \tgmodel, and \nbmodel\,-- can effectively capture different aspects of user behavior. 
 The competitive empirical performance of our approach compared to existing methods, particularly in scenarios with limited pilot data, suggests that the stable beta-scaled process prior effectively captures the power-law behavior often observed in user engagement data.
The choice of the model to use can be informed by two criteria. First, is the experimenter interested in statistics concerning the number of clicks, or do they focus solely on coarse daily-activity data? In the former case, the \nbmodel is warranted.
In the latter case, the experimenter should try to understand if a constant user-engagement propensity over time is a reasonable assumption. For instance, if a trigger in the A/B test is ``logging into my email account,'' we could expect that the assumption is reasonable, while if the trigger is ``making a specific purchase'' then such an assumption is extremely unlikely to hold in practice. If the constant user-engagement propensity assumption is met, then the \bemodel model leverages the whole dataset and can be expected to yield better forecasts. On the other hand, if the assumption is not met, the estimated parameters under the \bemodel model will be biased and predictive accuracy will suffer.
Alternatively, model fit can be assessed quantitatively by comparing the marginal log-likelihood across models on held-out pilot data, or by posterior predictive checks.

A limitation of the empirical-Bayes implementation is that posterior credible intervals condition on plug-in hyperparameter estimates. Consequently, their frequentist calibration is not guaranteed under model misspecification or when hyperparameter uncertainty is non-negligible. While our empirical emphasis in this paper is on point forecasts and duration-to-target accuracy, for applications where interval calibration itself is a primary decision criterion, fully Bayesian inference over hyperparameters or post-hoc calibration using held-out experiments would be natural extensions.

% The methodology introduced here extends beyond the specific context of online A/B testing. The flexible framework can be adapted to other domains where predicting user engagement is crucial, such as customer retention analysis and churn prediction. Moreover, our approach to handling both coarse and granular activity data provides a template for developing similar methodologies in other applications where data may be available at different levels of resolution.

This work suggests several directions for future research. First, extending the models to incorporate user covariates could provide more nuanced predictions based on user characteristics or experimental conditions. 
Second, our models treat each arm of the A/B test independently. However, the first triggering time of users is independent of the treatment, which affects only the re-trigger counts. This kind of dependence is not captured by standard nonparametric priors, such as the ones based on hierarchical or nested processes, and will be the object of future investigation.

% In conclusion, this work demonstrates that sophisticated Bayesian nonparametric methods can be made practically accessible without sacrificing theoretical rigor. The methods presented here provide experimenters with reliable tools for predicting user engagement and making informed decisions about experiment duration, while opening new avenues for theoretical research in Bayesian nonparametrics and their applications to large-scale experimentation.

\section*{Code and data availability}

Code, prepared public inputs, and complete reproduction instructions are
available at
\url{https://github.com/lorenzomasoero/OnlineActivityPredictionIBP}.
The package provides fixed-seed workflows and a result-to-command map for
the analyses on REES46, UCI, and ASOS. Running \texttt{make repro}
regenerates the core public-data fits, accuracy summaries, and figures;
separate documented targets cover the hitting-time and simulation analyses.

\FloatBarrier

\bibliographystyle{chicago}
\bibliography{references}

\FloatBarrier

\setcounter{section}{0}
\renewcommand{\thesection}{S\arabic{section}}
\renewcommand{\theHsection}{S\arabic{section}}

\setcounter{equation}{0}
\renewcommand\theequation{S\arabic{equation}}
\renewcommand\theHequation{S\arabic{equation}}

\renewcommand{\thetheorem}{\thesection.\arabic{theorem}}
\renewcommand{\theproposition}{\thesection.\arabic{proposition}}

\begin{center}
   \LARGE Supplementary material for:\\
    ``Online activity prediction via generalized Indian buffet process models''
\end{center}

\section*{Organization of the supplementary material}

The supplementary material is organized as follows.
Section \ref{app:background} provides the necessary background on Bayesian nonparametrics and states the main results used in our proofs.
Section \ref{app:proofs} collects the proofs of our theoretical results.
Section \ref{app:method_details} reports additional methodological details, including direct posterior sampling for the participation-threshold problem and numerical evidence for parameter estimation.
Section \ref{app:simu_true_model} reports the ``data from true model'' simulations moved from the main text.
The remaining sections collect generative schemes, details on competing methods, additional plots, benchmark construction details, data preprocessing details, and the power-formula derivation.

\section{Background Material on Bayesian nonparametrics and Random Measures}\label{app:background}

A completely random measure (CRM) on $\Omega$ (with associated $\sigma$-algebra) is a random element $\mu$ taking values in the space of (finite) measures over $\Omega$ (with associated Borel $\sigma$-algebra) such that for any $n$ and any collection of disjoint subsets $A_1, \ldots, A_n$ of $\Omega$, the random variables $\mu(A_1), \ldots, \mu(A_n)$ are independent.
It is well-known \citep{Kin67} that a CRM can be decomposed as a sum of a deterministic measure, an atomic measure with fixed atoms and random jumps, and an atomic measure where both atoms and jumps are random and form the points of a Poisson point process.
As customary in BNP, here we discard the deterministic parts of CRMs and consider measures of the form $\mu(\cdot) = \int_{\R_+} s  N(\dd s,\, \cdot) = \sum_{k \geq 1} \tau_k \delta_{\omega_k}(\cdot)$, where $ N = \sum_{k \geq 1} \delta_{(\tau_k, \omega_k)}$ is a Poisson random measure on $\R_+ \times \Omega$ with L\'evy intensity measure $\nu(\dd s, \dd x)$.
The L\'evy intensity characterizes the distribution of $\mu$ via the L\'evy-Khintchine representation of its Laplace functional. That is, for measurable $f: \Omega \rightarrow \R_+$
\[
    \E\left[e^{-\int_{\Omega} f(x) \mu (\dd x)}\right] = \exp \left\{- \int_{\R_+ \times \Omega} (1 - e^{- s f(x)}) \nu(\dd s \, \dd x) \right\}.
\]

Our focus is on homogeneous L\'evy intensity measures, namely measures of the form $\nu(\dd s, \, \dd x) = \kappa \rho(s) \dd s \, P_0(\dd x)$ where $\kappa > 0$ is a parameter, $P_0$ is a nonatomic probability measure on $\Omega$ and $\rho(s) \dd s$ is a measure on $\R_+$ such that $\int_{\R_+} \rho(\dd s) = +\infty$ and $\psi(u) := \int_{\R_+}(1 - e^{-us}) \rho(s) \dd s < +\infty$ for all $u>0$.
These conditions ensure that $0 < \mu (\Omega) < +\infty$ almost surely. We write $\mu \sim \mbox{CRM}(\kappa, \rho, P_0)$.
Under a trait allocation model, the law of $\mu$ provides a natural prior distribution for the parameter of the trait process. 

We recall below the main results due to \cite{Jam(17)} for the Bayesian analysis of trait allocations under a CRM prior.
Consider the trait allocations $X_j = \sum_{k \geq 1} X_{j, k} \delta_{\omega_k}$, $j=1, \ldots, n$. Given $\mu = \sum_{k \geq 1} \tau_k \delta_{\omega_k}$ we assume that
\begin{itemize}
    \item for each $k$, $X_{j, k} \mid \mu \iid G(\cdot \mid \tau_k)$ where $G(\cdot \mid s)$ is a distribution over the nonnegative integers and $\pi_G(s) := 1-G(0 \mid s)$ is the probability of positive activity
    \item the variables $X_{j, k}$ are independent across $k$, given $\mu$
    \item $\mu$ is a completely random measure with L\'evy intensity $\nu(\dd s \, \dd x) = \kappa \rho(s) \dd s P_0(\dd x)$.
\end{itemize}
We use the short-hand notation
\begin{equation}\label{eq:trp_crm}
    X_j = \sum_{k \geq 1} X_{j, k} \delta_{\omega_k}\mid \mu \iid \mbox{TrP}(G; \mu), \quad \mu \sim \mbox{CRM}(\kappa, \rho, P_0)  
\end{equation}
and further define
\[
    \varphi_i = \kappa \int_{\R_+} (1-\pi_G(s))^{i - 1}\pi_G(s) \rho(s) \dd s,
\]
so that $\sum_{i=1}^n\varphi_i=\kappa\int_{\R_+}\{1-(1-\pi_G(s))^n\}\rho(s)\dd s$.
\begin{proposition}[Marginal law, Proposition 3.1 in \cite{Jam(17)}]\label{prop:crm_marg}
    Let $X_1, \ldots, X_n$ be distributed as \eqref{eq:trp_crm}, such that the sample displays traits $\omega^*_1, \ldots, \omega^*_k$ and let $\mathcal B_j = \{i: X_i(\omega^*_j) > 0\}$, $|\mathcal B_j| = m_j$.
    Then, the marginal law of $(X_1, \ldots, X_n)$ is
    \[
        \exp\left\{
        -\kappa\int_{\R_+} [1-(1-\pi_G(s))^n]\rho(s)\dd s
        \right\}
        \prod_{j=1}^k
        \kappa\int_{\R_+}
        (1-\pi_G(s))^{n-m_j}
        \left\{\prod_{i\in\mathcal B_j}G(a_{i,j}\mid s)\right\}
        \rho(s)\dd s.
    \]
\end{proposition}

\begin{theorem}[Posterior law, Theorem 3.1 in \cite{Jam(17)}]\label{thm:crm_post}
    Let $X_1, \ldots, X_n$ be distributed as \eqref{eq:trp_crm}, such that the sample displays traits $\omega^*_1, \ldots, \omega^*_k$ and let $\mathcal B_j = \{i: X_i(\omega^*_j) > 0\}$, $|\mathcal B_j| = m_j$.
    Then, the posterior distribution of $\mu$ is equivalent to the distribution of
    \[
        \sum_{j=1}^k J^*_j \delta_{\omega^*_j} + \mu^\prime
    \]
    where
    \begin{enumerate}
        \item $J^*_j$ are independent positive random variables, also independent of $\mu^\prime$ with density
        \[
            f_j(s) \propto
            (1-\pi_G(s))^{n-m_j}
            \left\{\prod_{i\in\mathcal B_j}G(a_{i,j}\mid s)\right\}
            \rho(s)
        \]

        \item $\mu^\prime$ is a completely random measure with L\'evy intensity
        \[
            \kappa(1-\pi_G(s))^n\rho(s)\dd sP_0(\dd w).
        \]
    \end{enumerate}
\end{theorem}

\begin{proposition}[Compound Poisson representation, Proposition 3.3 in \cite{Jam(17)}]\label{prop:comp_poi}
    Let $X$ be as in \eqref{eq:trp_crm}. Then the law of $X$ is equivalent to the distribution of 
    \[
        \sum_{k=1}^K \tilde A_k \delta_{\tilde \omega_k}
    \]
    where $K \sim \mbox{Poi}(\varphi_1)$, $\tilde \omega_k \iid P_0$, and, independently, the $\tilde A_k$'s are i.i.d. random variables with values in $\{1,2,\ldots\}$ and probability mass function
    \[
        \prob(\tilde A_k=a)=\frac{\kappa}{\varphi_1}\int_{\R_+}G(a\mid s)\rho(s)\dd s,
        \qquad a=1,2,\ldots .
    \]
\end{proposition}

\section{Proofs}\label{app:proofs}

We will need the following technical lemma, which follows trivially from Lemma 1 in \cite{camerlenghi2022scaled}
\begin{lemma}\label{lemma:crm}
    Let $\mutilde \sim \mbox{SB-SP}(\alpha, c, \beta)$. Then, the law of $\mutilde \mid \Delta_{1, h}$ equals the one of a CRM with L\'evy intensity
    \[
        \Delta_{1, h}^{-\alpha} \alpha s^{-1-\alpha} I_{[0, 1]}(s) \dd s P_0(\dd x)
    \]
\end{lemma}

\begin{proof}[Proof of \Cref{thm:post_general}]
    The proof follows by combining \Cref{thm:crm_post} and \Cref{lemma:crm}.

    In particular, by standard disintegration arguments
    \[
        \mathcal L(\mutilde \mid Z_{1:D}) = \E[ \mathcal L(\mutilde \mid Z_{1:D}, \Delta_{1, h}) \mid Z_{1:D} ]
    \]
    and we can further introduce a positive-valued random variable $\tilde \Delta_{1,h}$ with law $ \mathcal L(\Delta_{1, h} \mid Z_{1:D})$ to obtain 
    \[
        \mathcal L(\mutilde \mid Z_{1:D}) = \int_{\R_+} \mathcal L(\mutilde \mid Z_{1:D},\tilde \Delta_{1,h}=\zeta) \mathcal L(\dd \zeta).
    \]

    Then, the term $\mathcal L(\mutilde \mid Z_{1:D}, \tilde \Delta_{1,h})$ can be obtained by observing that $\mutilde \mid \Delta_{1, h}$ is a CRM as in \Cref{lemma:crm} so that 
    $\mathcal L(\mutilde \mid Z_{1:D}, \tilde \Delta_{1,h})$ is the posterior distribution described in \Cref{thm:crm_post}.
    That is
    \[
        \mutilde \mid \tilde \Delta_{1,h} = \sum_{\useridx = 1}^{\usertotal_\daytotal} \theta^*_\useridx \delta_{\omega^*_\useridx} + \mu^\prime
    \]
    where $\theta^*_n\mid \tilde \Delta_{1,h}$ are independent with density 
    \[
        f_{\theta^*_\useridx}(\theta\mid \tilde \Delta_{1,h}) 
        \propto
        (1-\pi_{\mathcal G}(\theta))^{\daytotal-b_\useridx}
        \left\{\prod_{\dayidx \in B_\useridx} \mathcal G(Z_\dayidx(\omega^*_\useridx)\mid \theta)\right\}
        \theta^{-1-\alpha}
    \]
    which does not depend on $\tilde \Delta_{1,h}$.
    Moreover, $\mu^\prime \mid \tilde \Delta_{1,h}$ has L\'evy intensity
    \[
        \tilde \Delta_{1,h}^{-\alpha}\alpha(1-\pi_{\mathcal G}(\theta))^D
        \theta^{-1-\alpha}\mathbf 1_{(0,1)}(\theta)
        \dd\theta P_0(\dd\omega).
    \]

    It remains to obtain the law of $\tilde \Delta_{1,h}^{-\alpha} \distreq \Delta_{1, h}^{-\alpha} \mid Z_{1:D}$. An application of Bayes' rule yields that the density of $\tilde \Delta_{1,h}^{-\alpha}$ is proportional to
    \[
        x^{\usertotal_D}
        \exp\left\{
        -x\alpha\int_0^1\{1-(1-\pi_{\mathcal G}(\theta))^D\}
        \theta^{-1-\alpha}\dd\theta
        \right\}
        x^c e^{-\beta x},
    \]
    because $\Delta_{1,h}^{-\alpha}\sim\mathrm{Gamma}(c+1,\beta)$ under the SB-SP prior. Hence the posterior rate is
    \[
        \beta+\alpha\int_0^1\{1-(1-\pi_{\mathcal G}(\theta))^D\}
        \theta^{-1-\alpha}\dd\theta,
    \]
    and the posterior shape is $\usertotal_D+c+1$.
\end{proof}

\begin{proof}[Proof of \Cref{cor:post_models}]
    The proof follows by the definition of the Beta function and simple algebra.
\end{proof}

\begin{proof}[Proof of \Cref{prop:pred_unseen}]
Condition on $\tilde \Delta_{1,h}$ and use the residual CRM intensity from \Cref{cor:post_models}. For an unseen user to first trigger on absolute day $D_0+w$, it must have zero activity for the first $D_0+w-1$ days and positive activity on the next day. Integrating this probability against the residual L\'evy intensity gives
\[
    \tilde \Delta_{1,h}^{-\alpha}\alpha\int_0^1
    (1-\theta)^{r_*(D_0+w-1)}
    \{1-(1-\theta)^{r_*}\}
    \theta^{-1-\alpha}\dd\theta
    =
    \tilde \Delta_{1,h}^{-\alpha}\psi_{r_*}(D_0+w-1,1).
\]
Poisson process thinning gives the conditional Poisson laws and conditional independence over $w$. Therefore
\[
    \news{D_0}{D_1}
    =
    \sum_{w=1}^{D_1}N^*_{D_0+w}
    \mid \tilde \Delta_{1,h},data
    \sim
    \mathrm{Poi}\{\tilde \Delta_{1,h}^{-\alpha}\psi_{r_*}(D_0,D_1)\}.
\]
Integrating $\tilde \Delta_{1,h}^{-\alpha}\sim\mathrm{Gamma}(N_{D_0}+c+1,\beta+\psi_{r_*}(0,D_0))$ gives the stated negative-binomial law.
\end{proof}

\begin{proof}[Proof of \Cref{prop:pred_total}]

    The only non-trivial statement concerns the distribution of $\news{D_0}{D_1, j}$.
    Indeed, the representation follows by definition and the law of $ S_{D_0}^{(D_1)}$ follows directly by \Cref{cor:post_models}. Similarly, the Bayesian estimator for $\newstot{D_0}{D_1}$ is obtained by the linearity of expectations.

    To prove that 
    \[
        \news{D_0}{D_1, j} \mid Z_{1:D_0} \sim \mathrm{NegBin}(N_{D_0} + c + 1, p_{D_0}^{(D_1, j)}),
    \]
    we apply the argument in \citet[Theorem 2]{camerlenghi2022scaled}. Notice that the distribution of the new customers conditionally on the largest jump must be a Poisson distribution from the properties of CRMs.
    We can write from the predictive representation in \eqref{eq:pred_formal}
    \[
        \news{\pilotdays}{\followupdays,\freq}
        \mid Z_{1:\pilotdays},\tilde\Delta_{1,h}
        \overset{d}{=}
        \sum_{n\geq1}\mathbf 1\left\{
            \sum_{\dayidx=1}^{\followupdays}A'_{\pilotdays+\dayidx,n}=\freq
        \right\}.
    \]
    Here, conditionally on the residual CRM $\mu^\prime$, the variables $A'_{\pilotdays+\dayidx,n}$ are independent negative binomial random variables with parameters $r$ and $\theta^\prime_n$, where the $\theta^\prime_n$'s are the jumps of a Poisson point process with L{\'e}vy intensity
    \[
        \tilde{\lambda}(s)\dd s
        =\tail\tilde\Delta_{1,h}^{-\tail}(1-s)^{r\pilotdays}
        s^{-1-\tail}\mathbf 1_{(0,1)}(s)\dd s.
    \]
    Define
    \[
        S_{\pilotdays,\followupdays,n}
        :=\sum_{\dayidx=1}^{\followupdays}A'_{\pilotdays+\dayidx,n}.
    \]
    Conditionally on $\mu^\prime$, $S_{\pilotdays,\followupdays,n}\sim\mathrm{NegBin}(r\followupdays,\theta^\prime_n)$. Set
    \[
        q_\freq(s):=\binom{\freq+r\followupdays-1}{\freq}
        s^\freq(1-s)^{r\followupdays},
        \qquad
        \chi_\freq:=\tail\binom{\freq+r\followupdays-1}{\freq}.
    \]
    Write $U_\freq:=\news{\pilotdays}{\followupdays,\freq}$ for brevity. Consequently, averaging over the posterior residual Poisson point process conditionally on $\tilde\Delta_{1,h}$ gives
    \begin{align*}
        \E\left[t^{U_\freq}
        \mid Z_{1:\pilotdays},\tilde\Delta_{1,h}\right]
        &=\E\left[
        \left.\prod_{n\geq1}\{1+(t-1)q_\freq(\theta^\prime_n)\}
        \,\right|\,Z_{1:\pilotdays},\tilde\Delta_{1,h}\right]\\
        &=\exp\left\{-(1-t)\tilde\Delta_{1,h}^{-\tail}\chi_\freq
        \int_0^1(1-s)^{r(\pilotdays+\followupdays)}
        s^{\freq-\tail-1}\dd s\right\}\\
        &=\exp\left\{-(1-t)\tilde\Delta_{1,h}^{-\tail}\chi_\freq
        \mathrm B\left(\freq-\tail,r(\pilotdays+\followupdays)+1\right)\right\}\\
        &=\exp\left\{-(1-t)\tilde\Delta_{1,h}^{-\tail}
        \rho_{\pilotdays}^{(\followupdays,\freq)}\right\}.
    \end{align*}
    Here,
    \[
        \rho_{\pilotdays}^{(\followupdays,\freq)}
        :=\chi_\freq
        \mathrm B\left(\freq-\tail,r(\pilotdays+\followupdays)+1\right).
    \]
One can then integrate with respect to the posterior distribution of $\tilde \Delta_{1,h}^{-\alpha}$ to obtain the final result. Let $a=N_{\pilotdays}+c+1$ and $b=\beta+\psi_r(0,\pilotdays)$. Then
\begin{align*}
    \E\left[ t^{\news{\pilotdays}{\followupdays, \freq}} \mid Z_{1:\pilotdays}\right] &= 
    \left( \frac{b}{b + \rho_{\pilotdays}^{(\followupdays, \freq)} -t\rho_{\pilotdays}^{(\followupdays, \freq)}} \right)^a  \\
    &= \left( \frac{1-p_{\pilotdays}^{(\followupdays, \freq)}}{1-tp_{\pilotdays}^{(\followupdays, \freq)}} \right)^a
\end{align*}
for any $t<1/|p_{\pilotdays}^{(\followupdays, \freq)}|$, with 
\[
    p_{\pilotdays}^{(\followupdays, \freq)}:=\frac{\rho_{\pilotdays}^{(\followupdays,\freq)}}{\beta + \psi_r(0,\pilotdays)+\rho_{\pilotdays}^{(\followupdays, \freq)}} \le 1.
\]
This is the probability generating function of a negative binomial distribution under the convention in the main text:
\[
    \prob(\news{\pilotdays}{\followupdays,\freq} = \ell \mid Z_{1:\pilotdays}) = \binom{\ell + a-1}{\ell} (1-p_{\pilotdays}^{(\followupdays, \freq)})^a(p_{\pilotdays}^{(\followupdays, \freq)})^{\ell} \mathbf 1_{\{\ell \in \NN\}}.
\]
    
\end{proof}

\begin{proof}[Proof of \Cref{thm:marg_general}]

    With $\mathcal D_{D_0}$ denoting the model-specific pilot data defined in the theorem, the law of total expectation gives
    \[
        \prob(\mathcal D_{D_0}) = \E\left[\prob(\mathcal D_{D_0} \mid \Delta_{1,h}) \right],
    \] 
    where the argument inside the expectation is given by \Cref{prop:crm_marg}. In particular, it is straightforward to check that
    \[
        \prob(\mathcal D_{D_0} \mid \Delta_{1,h}) = (\alpha \Delta_{1,h}^{-\alpha})^{N_{D_0}} \exp\{-\psi_{r_*}(0, D_0)\Delta_{1,h}^{-\alpha}\} \prod_{n=1}^{N_{D_0}} \Theta_n.
    \]  
    The result then follows by marginalizing with respect to $\Delta_{1,h}^{-\alpha}\sim \mathrm{Gamma}(c+1, \beta)$.
\end{proof}

\section{Additional methodological details}\label{app:method_details}

\subsection{Estimating the number of days to reach a participation threshold}\label{app:est_days_direct}

We describe here the direct posterior simulation approach for $D_M$ referenced in \Cref{sec:est_days}.
Let $F^\prime = \sum_{\ell \geq 1} F^\prime_\ell \delta_{\omega^\prime_\ell}$ denote the measure tracking the absolute first triggering times of all unobserved users, i.e. users whose first trigger falls after day $D_0$. Equivalently,
$F^\prime(\omega) =\sum_{\ell \geq 1} F_\ell \mathrm{I}\{F_\ell > D_0\} \delta_{\omega_\ell}(\omega)$.
For each unobserved user, define the relative post-pilot waiting time $W^\prime_\ell = F^\prime_\ell - D_0$, and let $W^\prime_{(\ell)}$ be the $\ell$-th smallest value of these waiting times. Then $D_M = W^\prime_{(M-N_{D_0})}$, i.e., the $(M-N_{D_0})$-th earliest first-trigger time among unobserved users, measured in follow-up days after the pilot.
Unfortunately, the direct study of $W^\prime_{(M-N_{D_0})}$ is prohibitive: it can be proved that this quantity is an order statistic of a mixed Poisson process with non-constant rate, for which no closed form expression exists.
Nonetheless, we can simulate from $D_M$ and approximate its posterior via Monte Carlo.

Simulating directly from $W^\prime$ is cumbersome since it involves countably many random variables. Therefore, we define a suitable upper bound $D^{up}$ for $D_M$ and simulate the waiting times for those users that trigger within $D^{up}$ follow-up days.
The next theorem provides the key technical ingredient for this strategy.
\begin{theorem}\label{thm:marg_triggers}
    Let $W^\prime$ be defined as above and let $\tilde W = \sum_{\ell \geq 1} W^\prime_\ell \, \mathrm{I}\{W^\prime_\ell \leq D^{up} \} \, \delta_{\omega^\prime_\ell}$.
    Conditionally on $\tilde \Delta_{1,h}$ as in \Cref{cor:post_models}, $\tilde W$ admits the finite representation
    \[
        \tilde W \overset{d}{=} \sum_{w=1}^{D^{up}}\sum_{i=1}^{N^*_{D_0+w}} w \delta_{\omega^\prime_{w,i}},
    \]
    where the random variables $N^*_{D_0+w}$ are conditionally independent and
    \[
        N^*_{D_0+w} \mid \tilde \Delta_{1,h} \sim
        \mathrm{Poi}\{\tilde \Delta_{1,h}^{-\alpha}\psi_{r_*}(D_0+w-1,1)\},
        \qquad w=1,\ldots,D^{up},
    \]
    and $\omega^\prime_{w,i} \iid P_0$, with $r_*=1$ for the \bemodel and \tgmodel models and $r_*=r$ for the \nbmodel model.
    Therefore, the total number of post-pilot users represented in $\tilde W$ has the same marginal law as $\news{D_0}{D^{up}}$ in \Cref{prop:pred_unseen}. Conditional on this total count, the represented waiting times are iid with $\prob(W=w)=\psi_{r_*}(D_0+w-1,1)/\psi_{r_*}(D_0,D^{up})$, $w=1,\ldots,D^{up}$; for the \bemodel and \tgmodel models this probability is proportional to $\mathrm{B}(1-\alpha,D_0+w)$.
\end{theorem}

\begin{proof}[Proof of \Cref{thm:marg_triggers}]
The proof follows arguing as in \Cref{prop:pred_unseen}. Conditional on $\tilde \Delta_{1,h}$, Poisson thinning of the residual CRM implies that the number of unobserved users whose first post-pilot trigger occurs at follow-up day $w$ is Poisson with mean
\[
    \tilde \Delta_{1,h}^{-\alpha}\psi_{r_*}(D_0+w-1,1),
\]
and these counts are independent across $w=1,\ldots,D^{up}$.
Marking each such user by its follow-up waiting time gives the stated representation of $\tilde W$, with labels independently distributed as $P_0$.
Summing the Poisson counts over $w$ gives the same conditional law as $\news{D_0}{D^{up}}$ in \Cref{prop:pred_unseen}; conditional on this total count, the waiting-time probabilities are proportional to the corresponding Poisson means, which gives the stated probabilities after normalization.
\end{proof}

\begin{algorithm}[tp]
\caption{Posterior sampling for $D_M$}\label{algo:simulate_D_M}
\begin{algorithmic}
\STATE {\bfseries Input:} Observations, total-user target $M$, optional truncation $D^{up}$, number of Monte Carlo iterations $K$.
% \KwResult{$D_1, \ldots, D_K$.}
\STATE Set $m=M-N_{D_0}$. If $m\le 0$, return $D_1=\cdots=D_K=0$.

\FOR{$k = 1, \ldots, K$}
    \STATE Sample $\tilde \Delta_{1,h}^{-\alpha}$ from its posterior in \Cref{cor:post_models}.

    \STATE Set $C=0$ and $w=0$.

    \WHILE{$C<m$ and either no truncation is used or $w<D^{up}$}
        \STATE Set $w=w+1$.
        \STATE Sample $N^*_{D_0+w}\sim\mathrm{Poi}\{\tilde \Delta_{1,h}^{-\alpha}\psi_{r_*}(D_0+w-1,1)\}$.
        \STATE Set $C=C+N^*_{D_0+w}$.
    \ENDWHILE

    \STATE If $C\ge m$, set $D_k=w$; otherwise mark the draw as censored at $D^{up}$.
\ENDFOR
\STATE {\bfseries Return:} $D_1, \ldots, D_K$.
\end{algorithmic}
\end{algorithm}

\Cref{thm:marg_triggers} gives an explicit characterization of the posterior law of our inferential object $D_M$. In particular, building on that, we can simulate from the posterior
of $D_M$ using Algorithm \ref{algo:simulate_D_M}.
Without a fixed truncation, the sequential version gives exact posterior samples of $D_M$. If a fixed $D^{up}$ is used, the draw is censored when fewer than $M-N_{D_0}$ waiting times are sampled. The procedure does not require a burn-in period like Markov chain Monte Carlo.
However, \Cref{algo:simulate_D_M} can require a large number of samples to provide reliable estimates and can therefore be computationally expensive, especially in settings when $N_{D_0}$ is large, say in the order of the hundreds of thousands or larger.

\subsection{Empirics for parameter estimation}\label{app:param_empirics}

We provide here numerical evidence for both the maximum marginal likelihood and curve fitting approaches to parameter estimation.
We consider specifically the \nbmodel model, as the \bemodel can be recovered by setting $r = 1$.

We generate $D=365$ observations from the true model with fixed parameters $(\alpha,c,\beta,r)=(0.5,30,2,5)$.
\Cref{fig:synthetic_log_like} in Appendix \Cref{app:plots} shows the profile of the negative (log) likelihood function around the true value of the parameters. Specifically, in each subplot, we evaluate the negative log-likelihood in the right and left neighborhoods of the true value, shifting one coordinate at a time. We see that the negative log-likelihood is at least locally convex for all parameters. In these synthetic NB-SSP experiments, the derivative-free optimization routine consistently finds hyperparameter values with good predictive performance.

Next, we compare the performance of maximum marginal likelihood and curve-fitting on the unseen user prediction task. To assess the performance of these two alternative empirical Bayes strategies, we consider synthetic data drawn from the \nbmodel model.
For a fixed total duration $D = 500$, we draw samples $Z_{1:D} \sim \mbox{\nbmodel}$.
For each $D_0 \in \{5, 10, 20, 50\}$, we retain $Z_{1:D_0}$ as a training set, estimate the parameters $\alpha,c,\beta,r$ using either \Cref{eq:ml} or \Cref{eq:regression}, and compute the corresponding accuracy $v_{D_0}^{(D-D_0)}$ defined in \Cref{eq:acc_metric}.
We repeat this procedure $M=100$ times (by re-drawing each time $Z_{1:D}$ from the prior), and report the median accuracy and central 80\% empirical interval as $D_0$ increases, for three different choices of the hyperparameters $(\alpha,c,\beta,r)$ in \Cref{fig:synthetic_accuracy}.
While the confidence intervals greatly overlap, we can conclude that maximum marginal likelihood yields better predictive accuracy on average. Therefore, whenever possible, we advocate this approach. For ASOS, we fit TG-SSP by maximizing the geometric first-trigger marginal likelihood obtained from \Cref{eq:marg} via \Cref{eq:ml}. We fit IBP and NB-SSP using the curve-regression criterion in \Cref{eq:regression}, because per-user re-trigger counts are unavailable.

\begin{figure}[ht]
    \centering
    \includegraphics[width=\linewidth]{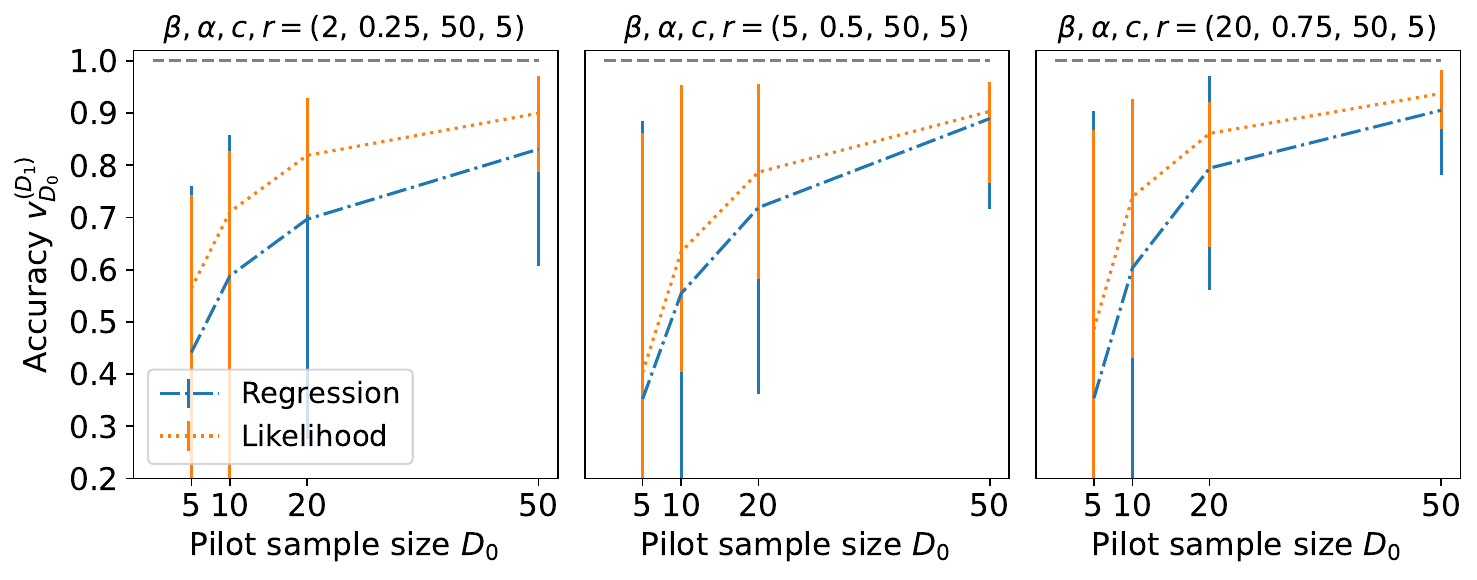}
    \caption{Prediction accuracy $v_{D_0}^{(D_1)}$ of the Bayesian nonparametric predictor on data from the model.}
    \label{fig:synthetic_accuracy}
\end{figure}

\Cref{fig:synthetic_accuracy} reveals that when the data is generated from the true model even extremely small sample sizes allow us to form accurate predictions of future user activity. This is important, as experimenters typically want to form predictions in the early phases of the experiment.

\section{Simulations from the true model}\label{app:simu_true_model}

\subsection{Binary activity data}\label{sec:simu1_binary}

We consider two data generating processes:
in the first (DG1) we simulate data from the \bemodel model using the generative scheme detailed in \Cref{app:ibp-scheme}.
In the second (DG2), we first simulate data as in DG1 and retain only the first trigger event $F_\useridx$ for each user $\useridx$. Then, for each user $\useridx$, we simulate $\varepsilon_\useridx \iid \mathcal U([0,0.5])$ and, independently across $\dayidx>F_\useridx$,
\[
    \daycount\mid(F_\useridx,\varepsilon_\useridx)
    \sim \mathrm{Be}\left(
        \varepsilon_\useridx\frac{1-\alpha}{1-\alpha+F_\useridx}
    \right).
\]
That is, in DG2, after the first trigger, users tend to be less active.

For both data generating processes, we fix $D_0=14$, $c = 2500$, $\beta = 0.5$, and simulate 50 independent datasets that differ only in the value of $\alpha \sim \mathrm{Beta}(4, 10)$.
Predictions are based on estimated parameters obtained by maximizing the marginal log-likelihood via \Cref{eq:ml}.
Because the activity data is binary, we compare only the \bemodel and \tgmodel models on the unseen user prediction task at $D_1 = 14$ follow-up days.
As a reference, we also include an \emph{oracle} predictor: the \bemodel predictive rule evaluated at the true data-generating parameters rather than at estimates.

\begin{figure}[H]
    \centering
    \includegraphics[width=\linewidth]{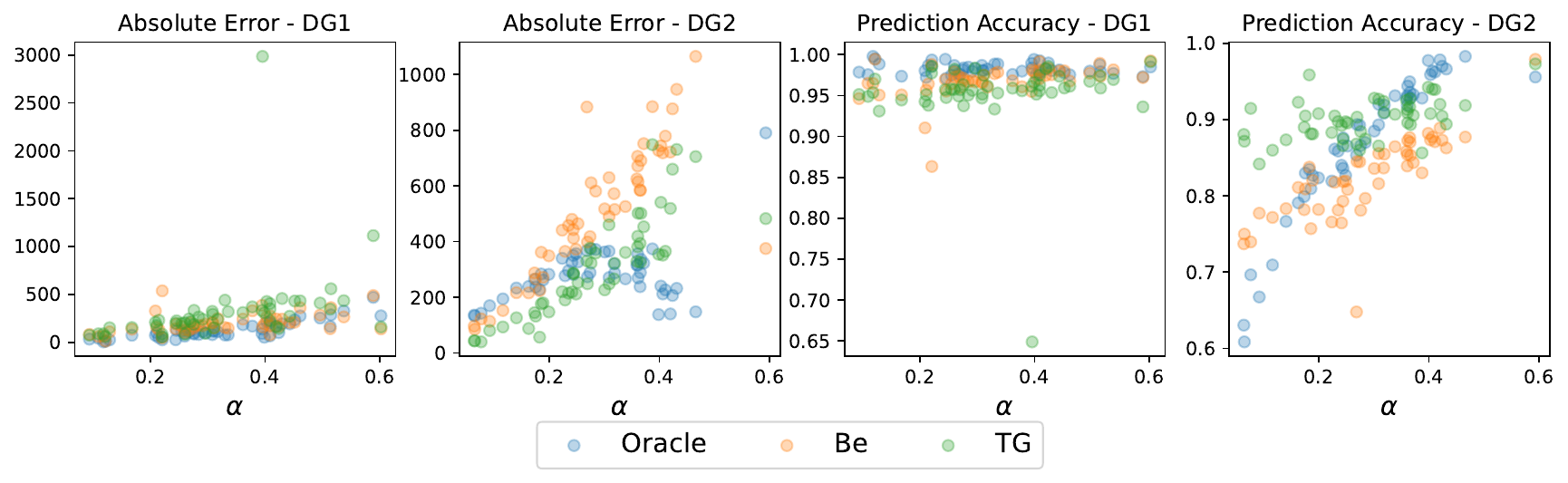}
    \caption{Absolute prediction error and accuracy for the simulated scenarios DG1 and DG2.}
    \label{fig:simu1_binary}
\end{figure}

In DG1, the oracle estimator achieves the smallest errors; \bemodel is competitive while \tgmodel yields slightly larger errors.
Under DG2, \tgmodel performs significantly better than \bemodel across all values of $\alpha$.
Notably, \tgmodel also outperforms the ``oracle'' under DG2. This is because the oracle is oracular only in its parameters: under DG2 the post-first-trigger daily activity departs from the \bemodel law it assumes, so conditioning on the daily-activity summary is misspecified. \tgmodel uses only the first-trigger times, which remain \bemodel-faithful under DG2, and is therefore correctly specified for this signal---so the oracle is not a genuine oracle in this setting.

\subsection{Count-valued activity data}

We draw synthetic data from the \nbmodel model with parameters $(\alpha,c,\beta,r)=(0.5,50,0.1,5)$ for $D=200$ days. We retain the first $\pilotdays=20$ days to fit hyperparameters and form predictions.

\begin{figure}[H]
    \centering
    \includegraphics[width=\linewidth]{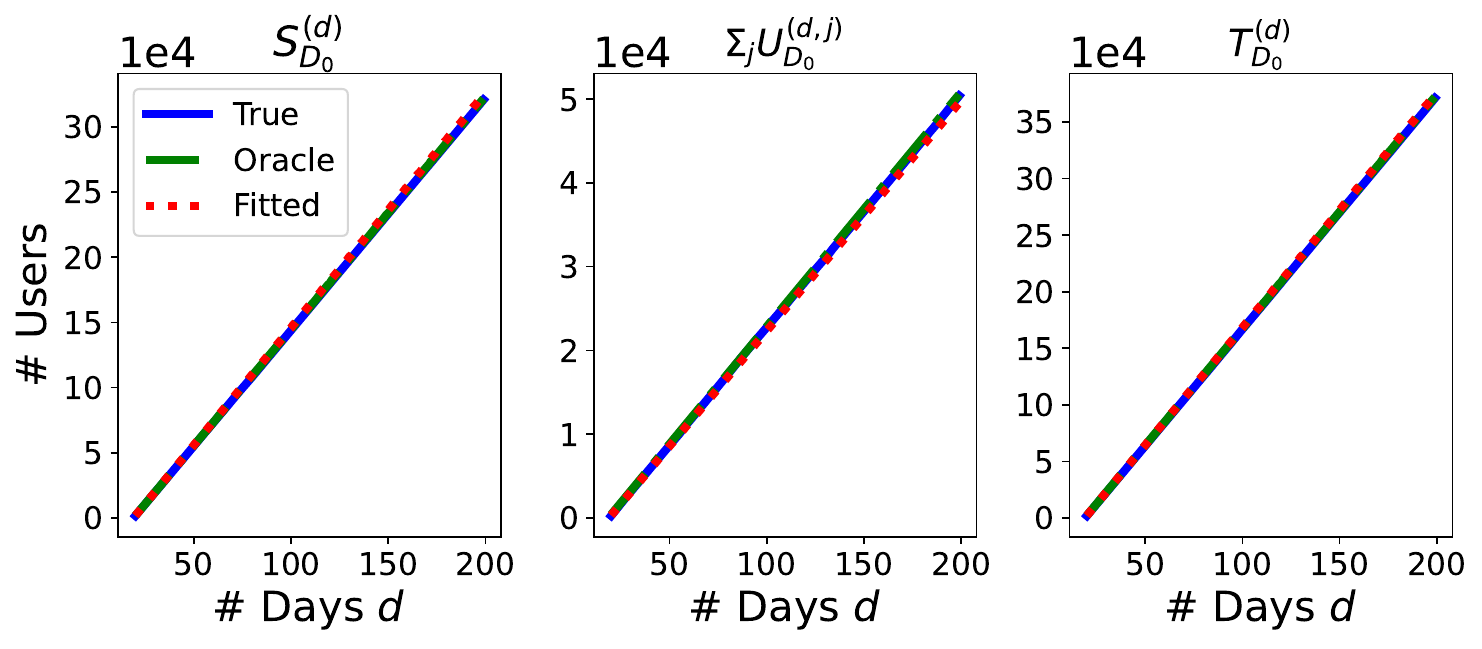}
    \caption{Prediction of future user activity from synthetic \nbmodel data.}
    \label{fig:synthetic_sums}
\end{figure}

\section{Generative Schemes under the Model}\label{app:ibp-scheme}

We describe below two generative schemes that can be thought of as generalizations of the IBP to our class of prior as well as to deal with non-binary scores. In particular, for the Bernoulli model this scheme is a straightforward consequence of Proposition 5 in \cite{camerlenghi2022scaled}. For the Geometric model, it is a rewriting of  \Cref{prop:comp_poi}.
For these schemes, write $\gamma(x,y):=\psi_1(x,y)$.

\subsection{Bernoulli Model}

\begin{enumerate}
    \item In the first experiment day,
    \[
        N_1
        \sim
        \mbox{NegBin}\left(
            c+1,
            \frac{\gamma(0,1)}{\beta+\gamma(0,1)}
        \right)
    \]
    users trigger. 

    \item After $d$ days suppose we have observed $N_d$ unique users $\omega^*_1, \ldots, \omega^*_{N_d}$, and let $b_i$ be the active-day count for observed user $i$.
    Then on day $d+1$, each previously seen user triggers independently with probability
    \[
        \prob(\text{user i active on day }d+1\mid data)
        =
        \frac{b_i-\alpha}{d+1-\alpha}.
    \]
    
    Moreover,
    \[
        N^*_{d+1}
        \sim
        \mbox{NegBin}\left(
            N_d+c+1,
            \frac{\gamma(d,1)}{\beta+\gamma(0,d+1)}
        \right)
    \]
    new users (i.e., previously unseen) will trigger for the first time.
\end{enumerate}

\subsection{Geometric Model}

The triggering times of users active during $\{1, \ldots, D^*\}$ are distributed as the random measure in \Cref{thm:marg_triggers} with $D_0=0$ and $D^{up}=D^*$. In particular,
\[
    N_{D^*}
    \sim
    \mbox{NegBin}\left(
        c+1,
        \frac{\gamma(0,D^*)}{\beta+\gamma(0,D^*)}
    \right).
\]
Conditional on $N_{D^*}$, the triggering times are i.i.d. random variables supported on $\{1, \ldots, D^*\}$ such that 
\[
    \prob(Y_\ell = y)
    =
    \frac{\gamma(y-1,1)}{\gamma(0,D^*)}
    =
    \frac{\alpha \mathrm{B}(1-\alpha,y)}{\psi_1(0,D^*)}.
\]

\section{Details on competing methods} \label{sec:app_competing}

To benchmark the performance of our newly proposed method, we consider a number of alternatives which have previously been proposed in the literature. We here provide additional details on these methods.

\subsection{Beta-binomial (BB) and hierarchical beta-geometric (HBG) predictors}

The beta-binomial and hierarchical beta-geometric models are finite-dimensional Bayesian generative models for trigger data which work by imposing a pre-determined upper bound $N_{\infty}$ on the number of units in the population, and assuming that for every unit $n=1,\ldots,N_{\infty}$ there exists a corresponding rate $\theta_n$ which governs the unit activity. In particular, these models assume
\[
    \theta_n \sim \mathrm{Beta}(\alpha,\beta).
\]
The beta-binomial model then assumes that we can observe for every day $\dayidx$ of experimentation and every unit $n$ whether the unit triggered in the experiment on that day, and postulates
\[
    Z_{\dayidx, n} \mid \theta_1,\ldots,\theta_{N_{\infty}} \sim \mathrm{Bernoulli}(\theta_n),
\]
i.i.d.\ across days $\dayidx$ for the same unit $n$, and independently across different units. The HBG model uses each unit's ``first trigger date'' $\dayidx_n$ and postulates
\[
    Z_{n} \mid \theta_1,\ldots,\theta_{N_{\infty}} \sim \mathrm{Geometric}(\theta_n).
\]
For the BB model, we use the estimator provided in \citet[Section 1]{ionita2009estimating} to produce $U_{\pilotdays}^{\followupdays}$. For the HBG model, we adopt the Monte Carlo approach devised in \citet[Section 3]{richardson22a} to produce the corresponding estimates. We provide code to fit these models and produce the corresponding estimates.

\subsection{Jackknife (J) predictors}

Jackknife estimators have a long history in the statistics literature, dating back to \citet{quenouille1956notes,tukey1958bias}. Here we consider the jackknife estimators developed by \citet{gravel2014predicting} who extended the work of \citet{burnham1979robust}. In particular, the $k$-th order jackknife is obtained by considering the first $k$ values of the resampling frequency spectrum; that is, by adequately weighting the number of users who appeared exactly $1, 2, \ldots, k$ times in the experiment. We adapt the code provided in \url{https://github.com/sgravel/capture_recapture/tree/master/software} to form our predictions.

\subsection{Good--Toulmin (GT) predictors}

Good--Toulmin estimators date back to the seminal work of \citep{good1956number}. Here, we adapt the recent approach of \citet{chakraborty2019using} for the problem of predicting the number of new genetic variants to be observed in future samples to online randomized experiments (in particular, we use the estimators provided in Equation (6) of the supplementary material). To form these predictions, we adapt to our setting the code provided in \url{https://github.com/c7rishi/variantprobs}.

\subsection{Linear programming predictors}

Linear programs have been used for rare event occurrence ever since the seminal work of \citet{efron1976estimating}. Here, we adapt to our setting the predictors proposed in \citet{zou2016quantifying} via the \texttt{UnseenEST} algorithm. We adapt the implementation provided by the authors at \url{https://github.com/jameszou/unseenest} to perform our experiments.

\section{Additional Figures}\label{app:plots}

Figure \ref{fig:synthetic_log_like} shows the evaluation of the log-likelihood for an \nbmodel model.

\begin{figure}[htbp]
    \centering
    \includegraphics[width=\linewidth]{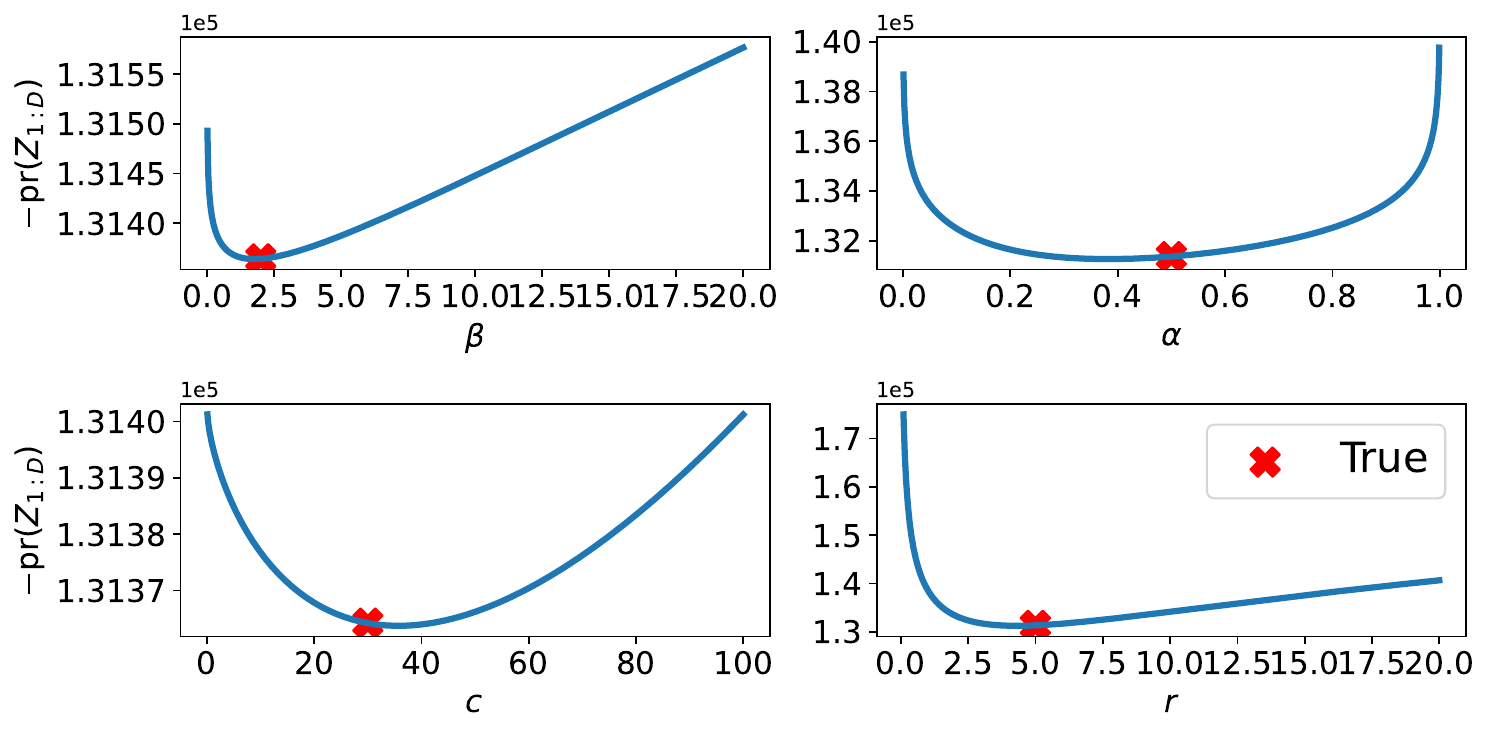}
    \caption{Evaluation of the log-likelihood in neighborhoods of the true values for data $Z_1, \ldots, Z_D$ from an \nbmodel  model with  $\alpha, \beta, c, r = (0.5, 2, 30, 5)$ and $D=365$
    }
    \label{fig:synthetic_log_like}
\end{figure}

Figure \ref{fig:synthetic_accuracy_sums} reports the accuracy for the \nbmodel predictor on Zipf data for different choices of the tail parameter $\tau$.

\begin{figure}[htbp]
    \centering
    \includegraphics[width=0.7\linewidth]{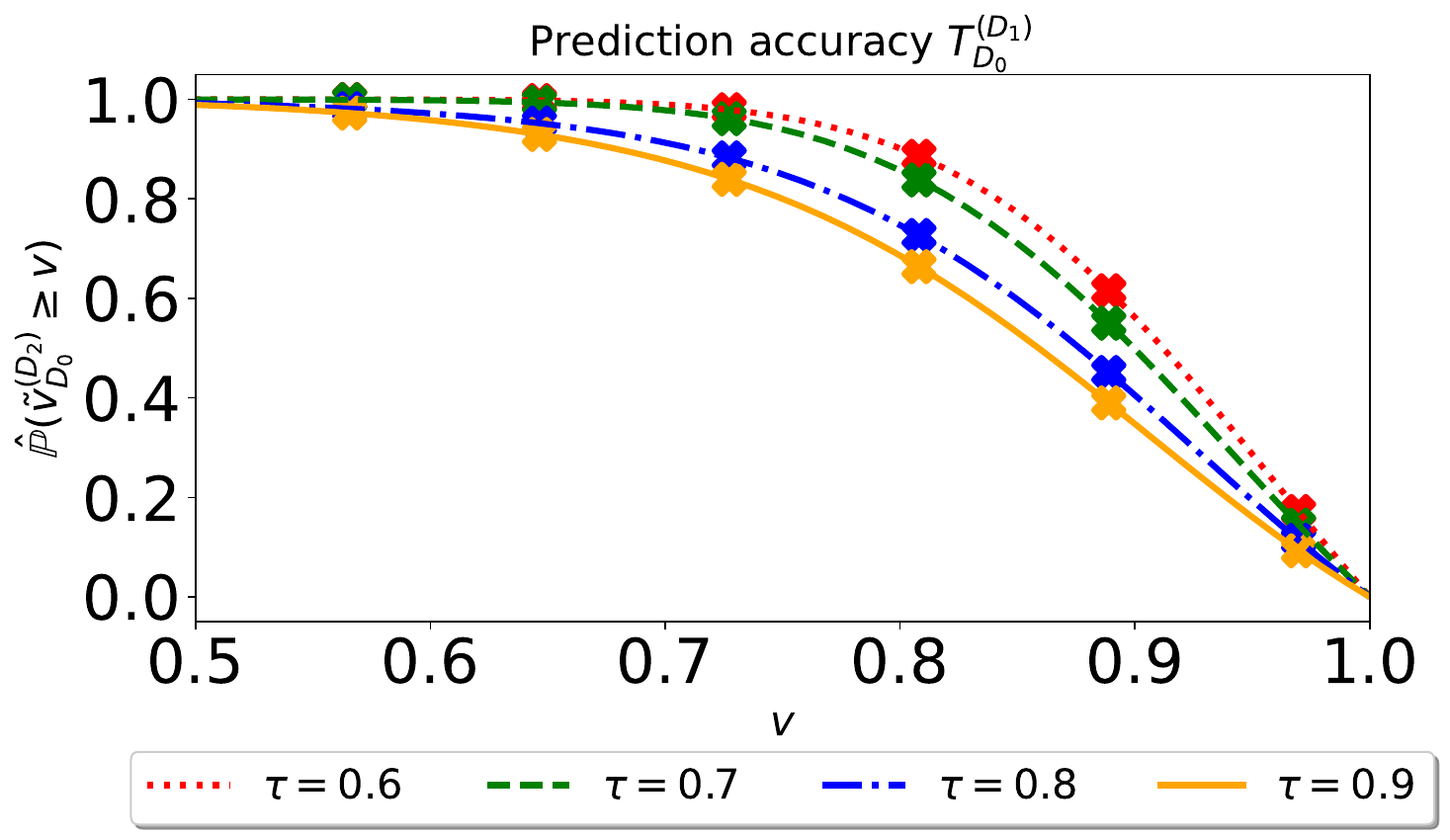}
    \caption{Prediction accuracy $\tilde{v}_{D_0}^{(D_1)}$ of the \nbmodel predictor for total re-trigger counts on Zipf data for different choices of the parameter $\tau$.
    }
    \label{fig:synthetic_accuracy_sums}
\end{figure}

Figure \ref{fig:true_accuracy_sums} reports the accuracy for the total triggering activity of the \nbmodel predictor on proprietary data.

\begin{figure}[htbp]
    \centering
    \includegraphics[width=0.65\linewidth]{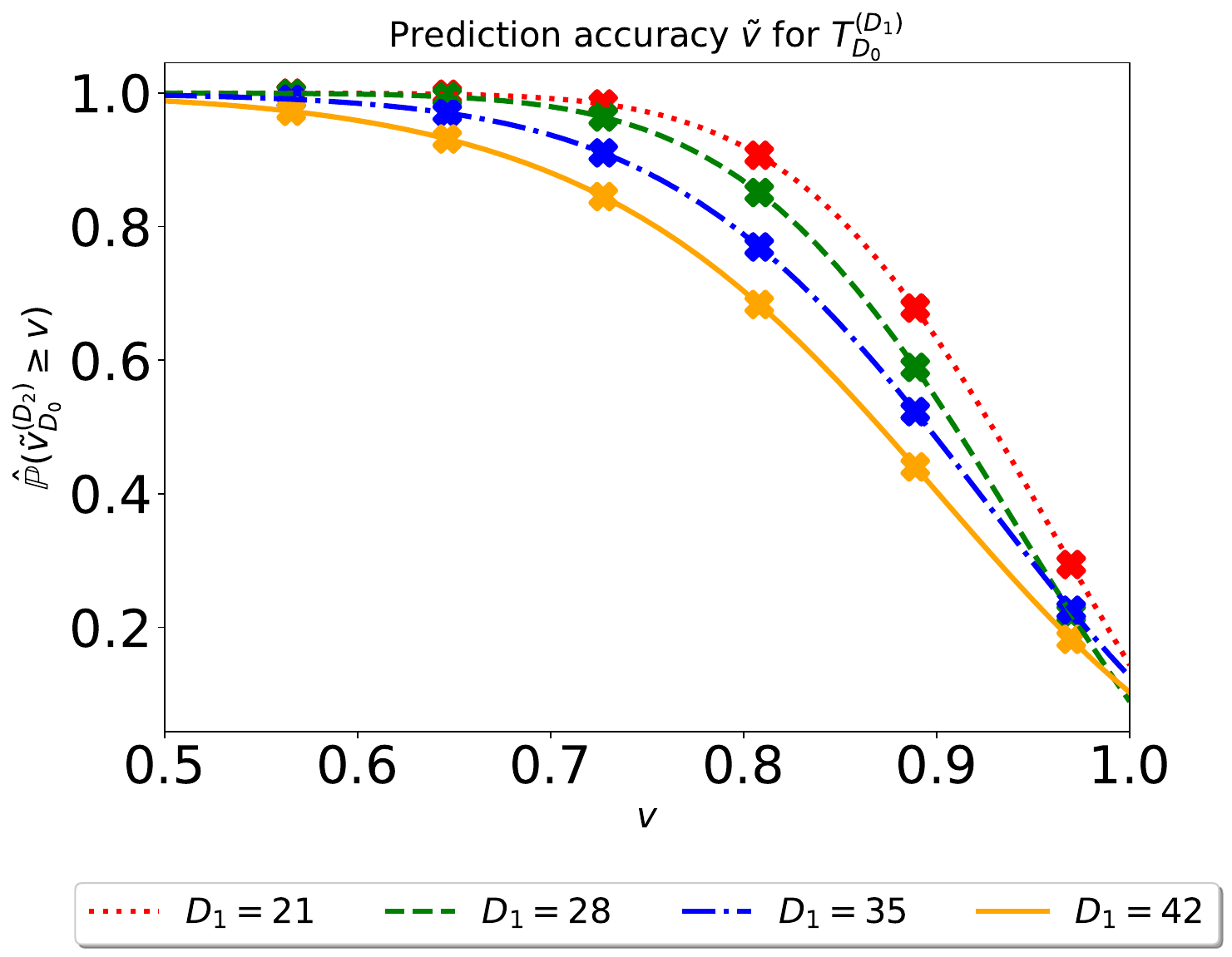}
    \caption{Survival plot for the prediction accuracy $\tilde{v}_{\pilotdays}^{(\followupdays)}$ of the \nbmodel predictor for the total triggering activity (proprietary data). For a given accuracy level (horizontal axis), we plot the fraction of experiments achieving at least that level of accuracy (vertical axis) for different extrapolation values $D_1$.
    }
    \label{fig:true_accuracy_sums}
\end{figure}

Figures~\ref{fig:app_uci}--\ref{fig:app_asos} show prediction trajectories for all experiment windows, complementing the three case studies in the main text (\Cref{fig:case_uci}).

\begin{figure}[htbp]
    \centering
    \includegraphics[width=\linewidth]{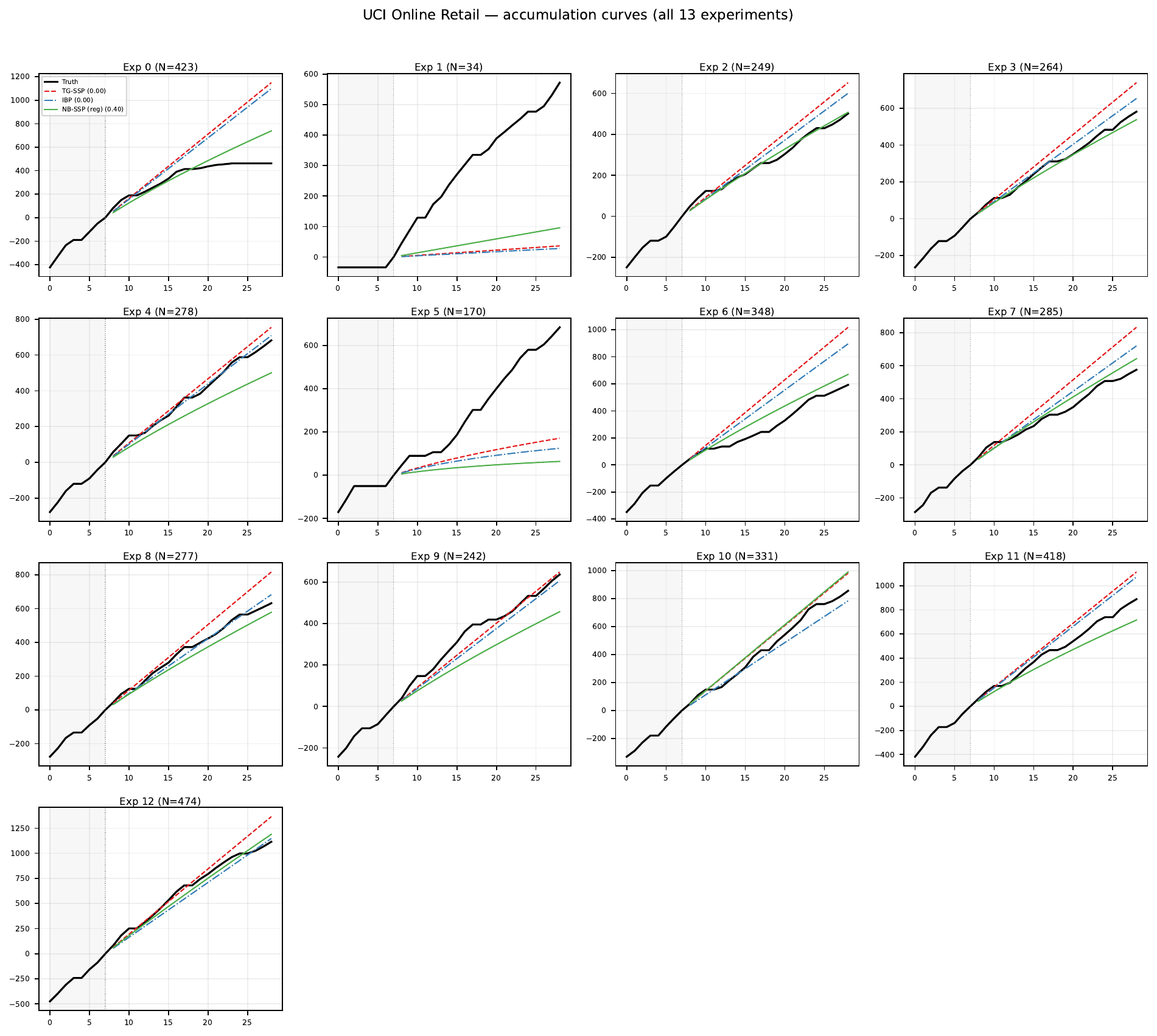}
    \caption{UCI Online Retail: accumulation curves for all 13 experiment windows. Black = truth; colored lines = method predictions. Gray region = pilot ($\pilotdays = 7$).}
    \label{fig:app_uci}
\end{figure}

\begin{figure}[htbp]
    \centering
    \includegraphics[width=\linewidth]{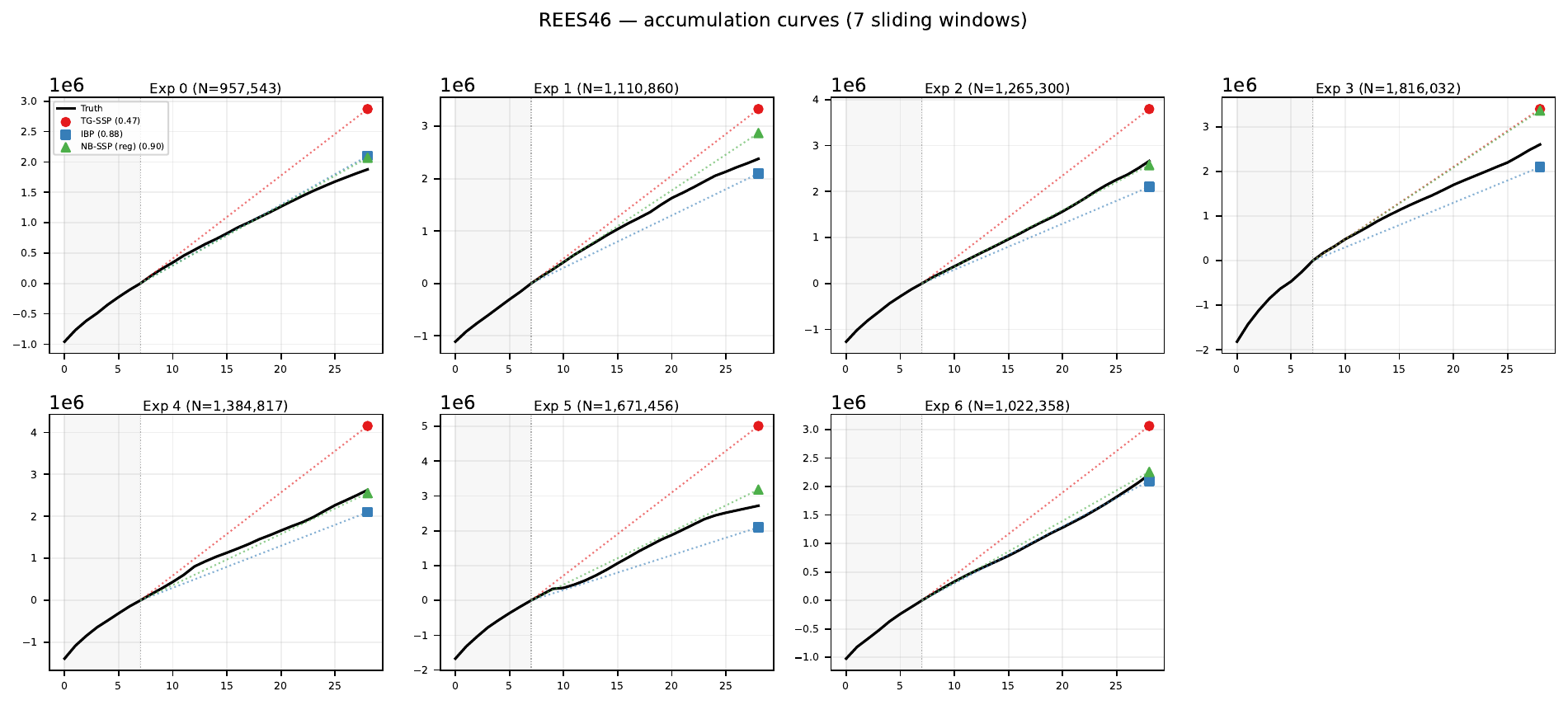}
    \caption{REES46: accumulation curves for all 7 non-overlapping 28-day windows. Markers at day $D$ show final predictions (full trajectories not available for this dataset).}
    \label{fig:app_rees46}
\end{figure}

\begin{figure}[htbp]
    \centering
    \includegraphics[width=\linewidth]{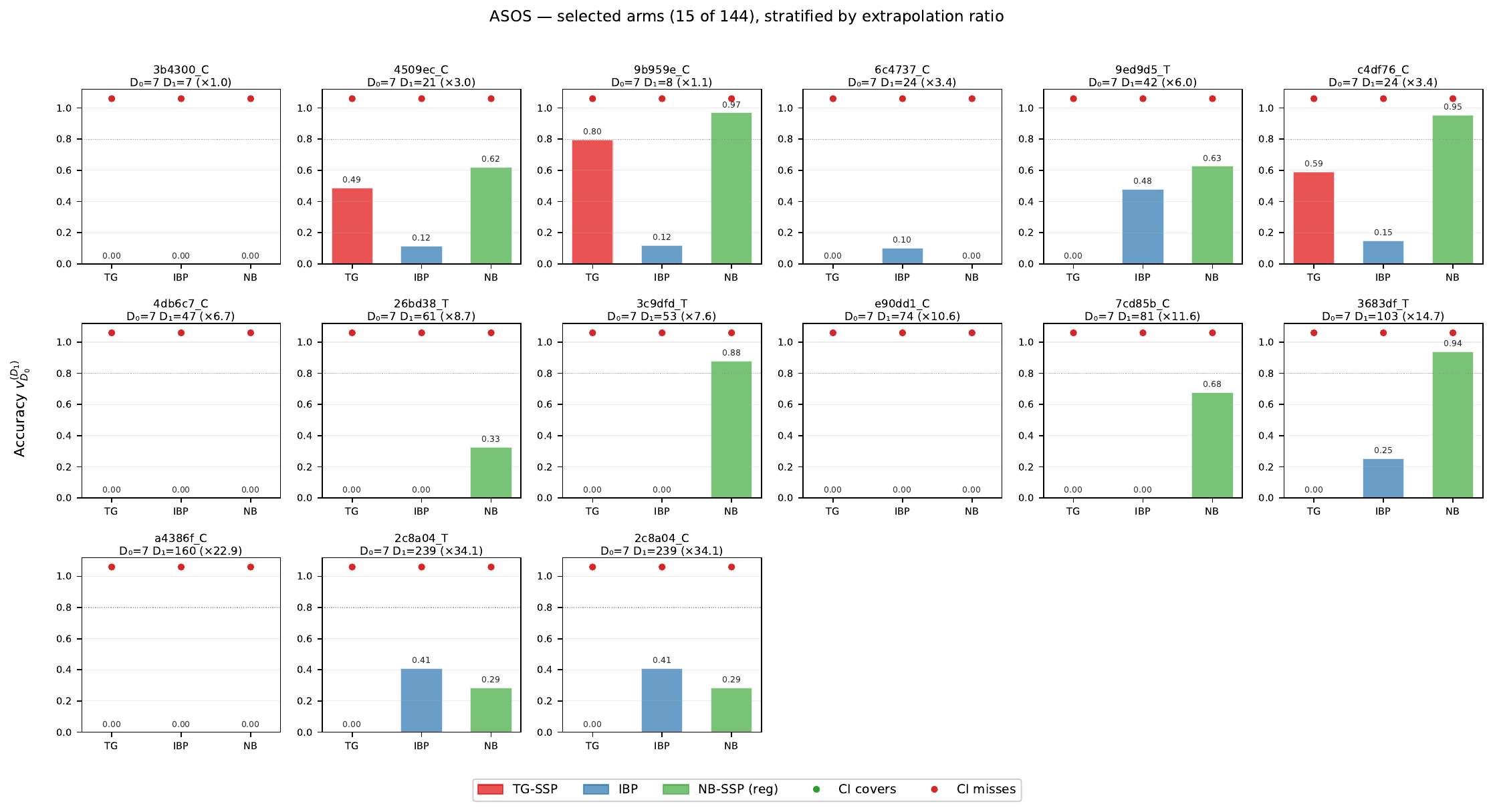}
    \caption{ASOS: per-experiment accuracy for selected arms stratified by extrapolation ratio.}
    \label{fig:app_asos}
\end{figure}
\FloatBarrier

\section{Benchmark Construction Details}\label{app:benchmark_construction}

This section details the rolling-window benchmark construction used in Section~\ref{sec:exp_real}.

\subsection{Rolling windows with varying follow-up horizons}

For the REES46 dataset, we construct rolling experiment windows parameterized by the follow-up length $k \in \{21, 50, 100\}$ days.
Each window consists of a fixed pilot period of $\pilotdays = 7$ days followed by a follow-up of $\followupdays = k$ days, giving a total window length of $\pilotdays + k$ days.
Windows are placed at every possible start day (stride~1): for a dataset spanning $D_{\mathrm{total}}$ days, this yields $D_{\mathrm{total}} - \pilotdays - k + 1$ experiments per value of $k$.
The two dates without recorded events are retained as zero-event days, giving the complete 213-day calendar grid used to construct the rolling windows.

% Counts below are generated from the committed rolling-window artifacts
% using the public reproduction package.
On the 7-month REES46 dataset ($D_{\mathrm{total}} = \DtotalRees$ days), this produces:
\begin{itemize}
    \item $k = 21$: \NexpKtwentyOne rolling experiments (window = 28 days),
    \item $k = 50$: \NexpKfifty rolling experiments (window = 57 days),
    \item $k = 100$: \NexpKhundred rolling experiments (window = 107 days).
\end{itemize}

For each experiment, the ground-truth quantities are:
\begin{itemize}
    \item $N_{\pilotdays}$: number of distinct users whose first activity falls within the pilot (days $t, \ldots, t + \pilotdays - 1$);
    \item $\news{\pilotdays}{\followupdays}$: number of users whose first activity falls in the follow-up (days $t + \pilotdays, \ldots, t + \pilotdays + k - 1$).
\end{itemize}

The rolling construction ensures that each value of $k$ tests the model at a different extrapolation ratio $\followupdays / \pilotdays \in \{3, 7.1, 14.3\}$, providing a systematic assessment of how prediction accuracy degrades with horizon length.

\subsection{Non-overlapping windows}

For the primary accuracy comparison, we also use 7 non-overlapping 28-day windows ($\pilotdays = 7$, $\followupdays = 21$) from REES46 and 13 non-overlapping 28-day windows from UCI Online Retail.
These provide independent experiments (no shared data between windows) at the cost of fewer total experiments.

\subsection{Optimization landscape of TG-SSP on the 7-window benchmark}\label{app:tg_optimization}
The REES46 28-day benchmark comprises only 7 non-overlapping windows, and on it the TG-SSP median accuracy is sensitive to the random restarts of the geometric marginal-likelihood optimizer. Across 50 random seeds the per-seed median accuracy is bimodal (\Cref{fig:tg_bimodal}): 37 of 50 seeds ($74\%$) fall in a low mode (mean $0.47$) and 13 of 50 ($26\%$) in a high mode (mean $0.58$), with an overall median of $0.47$. The effect is per-window: each window's marginal likelihood has a dominant low optimum and a high optimum reached on only $2$--$10\%$ of seeds, so the high overall mode arises only when several windows land in their high optima simultaneously. We therefore report the typical (low) mode, $0.47$, in \Cref{tab:cross_accuracy}; the IBP ($0.81$) and \nbmodel ($0.90$) accuracies on the same benchmark are stable across seeds. This sensitivity is specific to the small 7-window benchmark; on the 186-window $k=21$ benchmark all methods are stable.
\begin{figure}[h]
    \centering
    \includegraphics[width=0.6\linewidth]{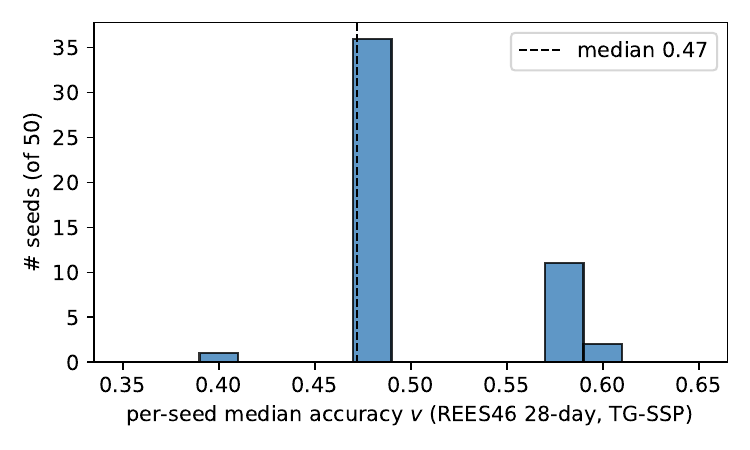}
    \caption{Per-seed median TG-SSP accuracy on the REES46 28-day benchmark across 50 seeds: bimodal, with a dominant low mode near $0.47$.}
    \label{fig:tg_bimodal}
\end{figure}

\section{Data Preprocessing Details}\label{app:data_preprocessing}

This section provides the details needed to reproduce the benchmark datasets used in Section~\ref{sec:exp_real}.
Code and preprocessed data are available at \url{https://github.com/lorenzomasoero/OnlineActivityPredictionIBP}.

\subsection{REES46 eCommerce Behavior Dataset}

The REES46 dataset is publicly available at \url{https://data.rees46.com/datasets/marketplace/}.
It consists of 7 monthly CSV files (October 2019 through April 2020), totaling approximately 380M events and 47\,GB uncompressed.
Each event record contains a timestamp (\texttt{event\_time}), event type (\texttt{view}, \texttt{cart}, or \texttt{purchase}), and user identifier (\texttt{user\_id}).
We define a \emph{trigger} as any event by a user on a given calendar day.
The daily trigger count $A_{d,n}$ for user $n$ on day $d$ is the number of events recorded for that user on that day.
Aggregating across all months produces approximately 15.6M unique users over an interval spanning 213 calendar days, with events recorded on 211 dates.

For each experiment window of length $D = \pilotdays + \followupdays$:
\begin{itemize}
    \item $N_{\pilotdays}$: number of distinct users whose first activity falls within the pilot (days $1, \ldots, \pilotdays$);
    \item $\news{\pilotdays}{\followupdays} = N_{\mathrm{total}} - N_{\pilotdays}$: number of users whose first activity falls in the follow-up;
    \item $\newstot{\pilotdays}{\followupdays}$: total trigger count across all users during the follow-up period.
\end{itemize}

\subsection{UCI Online Retail Dataset}
The UCI Online Retail dataset \citep{chen2012data} is available from the UCI Machine Learning Repository (ID~352).
It contains 541,909 transaction line items from a UK-based online retailer (December 2010 -- December 2011).
We retain rows with a valid \texttt{CustomerID}, positive \texttt{Quantity}, and non-cancelled invoices (excluding invoice numbers prefixed with ``C''), yielding 397,924 rows across 4,339 customers.
The daily trigger count $A_{d,n}$ is the number of transaction line items for customer $n$ on calendar day $d$.

\subsection{ASOS Dataset}
The ASOS Digital Experiments Dataset \citep{liu2021datasets} is available at \url{https://osf.io/64jsb/}.
It provides cumulative first-trigger counts per day for 72 experiments (144 treatment arms).
No additional preprocessing is required.

% {\color{purple}
\section{Power Formula Derivation}\label{app:power_derivation}

We provide a brief derivation of the power formula used in Section~\ref{sec:exp_real} to connect participation targets to statistical power.
For a comprehensive treatment of power calculations in A/B testing, see \citet{Gualavisi2025}; for a detailed derivation of the two-sample test power formula in a related Bayesian nonparametric context, see \citet{masoero2021bayesian}.

Consider a two-arm randomized experiment with equal allocation: $n = M_\eta / 2$ users per arm.
Let $\bar{Y}_T$ and $\bar{Y}_C$ denote the sample means in the treatment and control arms, respectively.
We test $H_0: \mu_T - \mu_C = 0$ against $H_1: \mu_T - \mu_C \neq 0$ at significance level $\alpha_{\mathrm{test}}$ (two-sided).
Under equal within-arm variance $\sigma^2$, the two-sample $z$-statistic is
\[
    Z = \frac{\bar{Y}_T - \bar{Y}_C}{\sigma \sqrt{2/n}} = \frac{\bar{Y}_T - \bar{Y}_C}{\sigma \sqrt{4 / M_\eta}}.
\]
Under $H_0$, $Z \sim \mathcal{N}(0, 1)$. We reject $H_0$ when $|Z| > z_{1-\alpha_{\mathrm{test}}/2}$.
Under $H_1$, let $\delta=|\mu_T-\mu_C|$ and $\lambda=\sqrt{M_\eta}\,\delta/(2\sigma)$. By symmetry, we may take $Z\sim\mathcal N(\lambda,1)$, and hence
\[
    \text{power}
    = \Phi\!\left(\lambda-z_{1-\alpha_{\mathrm{test}}/2}\right)
      + \Phi\!\left(-\lambda-z_{1-\alpha_{\mathrm{test}}/2}\right)
    \approx \Phi\!\left(\frac{\sqrt{M_\eta}}{2}\frac{\delta}{\sigma}-z_{1-\alpha_{\mathrm{test}}/2}\right),
\]
where $\Phi$ is the standard normal CDF and the approximation drops the typically negligible lower-tail rejection probability.

% Given a desired power $1 - \beta_{\mathrm{test}}$, the required total sample size is
% \[
%     M_\eta = \frac{4\sigma^2}{\delta^2} \left(z_{1-\alpha_{\mathrm{test}}/2} + z_{1-\beta_{\mathrm{test}}}\right)^2.
% \]
% This is the target that the experimenter sets; our hitting-time prediction $\hat{D}_\eta$ estimates how many days are needed to reach $M_\eta$ users.
% The mean absolute error $|\hat{D}_\eta - D_\eta|$ thus directly quantifies the accuracy of power-aware duration planning.

\end{document}